\documentclass[lettersize,journal]{IEEEtran}

\IEEEoverridecommandlockouts
\usepackage{cite}
\usepackage{amsmath,amssymb,amsfonts}
\usepackage{algorithmic}
\usepackage{float}
\usepackage{graphicx}
\usepackage{textcomp}
\usepackage{xcolor}
\usepackage{caption}
\usepackage{subcaption}
\usepackage{array}
\usepackage{multicol}
\usepackage{multirow}
\usepackage{makecell}
\usepackage{soul}
\usepackage{diagbox}
\allowdisplaybreaks

\newtheorem{corollary}{Corollary}
\newtheorem{proposition}{Proposition}

\newtheorem{lemma}{Lemma}
\newtheorem{theorem}{Theorem}
\newenvironment{proof}{{\ \noindent\it Proof:\ }}{\hfill $\square$\par}
\def\BibTeX{{\rm B\kern-.05em{\sc i\kern-.025em b}\kern-.08em
  T\kern-.1667em\lower.7ex\hbox{E}\kern-.125emX}}
\begin{document}

\title{Joint Coding-Modulation for Digital Semantic Communications via Variational Autoencoder
\author{
Yufei~Bo, Yiheng~Duan, Shuo~Shao, Meixia~Tao}
\thanks{Y. Bo, Y. Duan, and M. Tao are with the Department of Electronic Engineering and Shanghai Key Laboratory of Digital Media Processing and Transmission, Shanghai Jiao Tong University, Shanghai, 200240, China. (emails: \{boyufei01, duanyiheng, mxtao\}@sjtu.edu.cn)}
\thanks{S. Shao is with the School of Cyber Science and Engineering, Shanghai Jiao Tong University, Shanghai, 200240, China. (email: shuoshao@sjtu.edu.cn)}
\thanks{
Preliminary results of this work are presented in\cite{confVersion}, where the digital modulation scheme is fixed as BPSK.}
}
\maketitle



\begin{abstract}
Semantic communications have emerged as a new paradigm for improving communication efficiency by transmitting the semantic information of a source message that is most relevant to a desired task at the receiver. 
Most existing approaches typically utilize neural networks (NNs) to design end-to-end semantic communication systems, where NN-based semantic encoders output continuously distributed signals to be sent directly to the channel in an analog fashion. 
In this work, we propose a joint coding-modulation (JCM) framework for digital semantic communications by using variational autoencoder (VAE). 
Our approach learns the transition probability from source data to discrete constellation symbols, thereby avoiding the non-differentiability problem of digital modulation. 
Meanwhile, by jointly designing the coding and modulation process together, we can match the obtained modulation strategy with the operating channel condition. 
We also derive a matching loss function with information-theoretic meaning for end-to-end training. 
Experiments on image semantic communication validate the superiority of our proposed JCM framework over the state-of-the-art quantization-based digital semantic coding-modulation methods across a wide range of channel conditions, transmission rates, and modulation orders.
Furthermore, its performance gap to analog semantic communication reduces as the modulation order increases while enjoying the hardware implementation convenience.

\end{abstract}


\section{Introduction}

Semantic communications are emerging as a beyond-Shannon-type communication paradigm and are envisioned to be a key technology in future 6G wireless networks\cite{kalfa2021survey1, luo2022survey2, lan2021survey3, gunduz2022beyond, qin2021semantic}. 
Unlike traditional communications which deal with transmitting a source message to be exactly or approximately reconstructed at the destination, semantic communications focus on extracting and transmitting the semantic information of the source so as to enable the destination to make the right decision and execute the desired task. 
It can hence significantly improve the performance and efficiency of data transmission especially when the communication resources are limited.
Semantic communications are promising to enable a variety of intelligent services such as extended reality, metaverse, smart surveillance, and robotic collaboration.

By leveraging the advances in deep learning, semantic communications often use neural networks (NNs) to extract and encode the semantic information and can transmit various types of sources, such as speeches\cite{weng2021speech, han2023speech}, texts\cite{xie2021nn1,zhou2021semantic,xiao2022text}, images\cite{zhang2021bottleneck,liu2021semantics,zhang2022task,ntc}, videos\cite{wang2022wireless,jiang2022wireless} as well as multi-modal data\cite{xie2021nn2}. 
Specifically, the Transformer\cite{vaswani2017attention} architecture is typically adopted for processing text data\cite{xie2021nn1}, whereas convolutional neural networks (CNNs) are commonly seen in the transmission of images\cite{liu2021semantics, gunduz2019jscc1} or videos\cite{wang2022wireless}.
Generative adversarial networks (GANs)\cite{goodfellow2020generative} are also a common architecture for image data\cite{huang2022toward, dai2022gan}.
All in all, through the usage of NNs, semantic communications are suitable for various intelligent tasks such as classification\cite{zhang2021bottleneck}, translation\cite{xie2020benchamrk1}, and object recognition\cite{lee2019deep}.

However, using NN-based semantic encoders also makes it difficult to deploy semantic communication systems on modern digital communication devices. 
Most existing semantic communication system designs, especially those with end-to-end framework, utilize analog modulation instead\cite{xie2021nn1, weng2021speech,zhou2021semantic, xie2021nn2,zhang2021bottleneck,xie2022task, han2023speech, dai2022gan}.
That is, the NN-based semantic encoders usually output continuously distributed signals, which are directly sent into the communication channel without being modulated into discrete constellation symbols. 
Such analog transmission approach, however, is difficult to implement in practice, due to the non-ideal characteristics of hardware components including power amplifier. 
Thus, considering the complexity of hardware implementation as well as the compatibility with existing
digital communication protocols, to
design digital semantic communication systems now becomes a key for enabling the smooth shifting from traditional communication paradigm to semantic communications.

In spite of the undoubted importance of digital semantic communication, it is seldom considered in existing NN-based semantic communication systems due to its intrinsic mechanism. 
Specifically, the process of digital modulation requires a mapping from the continuous real-world source data to discrete constellation symbols, which is equivalent to a non-differentiable function\cite{unconstrained}. 
Therefore, NNs cannot achieve such functions due to the stochastic gradient descent (SGD) algorithms they use for optimization.
Though some existing works can realize digital modulation in semantic communications, they do not match the modulation process with channel states, which results in some performance loss. 
For example, in \cite{xie2020benchamrk1}, the continuous outputs of the NN encoder are first uniformly quantized into discrete number sequences, then modulated into constellation symbol sequences.
As another example, Jiang \emph{et al}. \cite{yeli2022harq} train a multi-layer perceptron (MLP) network with the Sigmoid function to get a non-uniform quantizer that has the minimum quantization error, and further use BPSK to transmit these bits.
As can be seen, the quantization process in the aforementioned works\cite{yeli2022harq, xie2020benchamrk1} is conducted separately without considering the channel condition, which results in unsatisfactory performance.
On the other hand, the work \cite{unconstrained} optimizes coding and modulation jointly through a hard-to-soft quantizer. It uses hard-decision quantization in the forward pass to map the NN output to the nearest constellation symbol. 
This non-differentiable operation is approximated by the Softmax function in the backward pass, enabling end-to-end training.
However, there is still room for improvement since hard decisions lead to information loss and a decline in performance.
Overall, the problems of the non-differentiability of the digital modulation process and the mismatch between digital modulation and channel states still remain to be solved in realizing digital modulation in semantic communications.

To tackle the above challenges in NN-based digital semantic communications, we propose a joint coding-modulation (JCM) framework as a novel end-to-end design for digital semantic communication based on variational autoencoder (VAE).
This framework comprises two main design methodologies. 
First, we address the issue of the non-differentiablility of the digital modulation process by using the NNs at the transmitter to learn the likelihood of constellation symbols instead of generating deterministic constellation symbols. 
More specifically, we first train NNs to learn the optimal transition probability from the source data to the constellation symbols. Then, we randomly generate the constellation symbol sequence based on this transition probability and send it to the receiver.
Second, we address the issue of the mismatch between digital modulation and channel states by designing a joint training process that combines coding, modulation and decoding together. 
More specifically, 
the NN at the transmitter and the NNs at the receiver are jointly trained with channel noise taken into consideration. This approach enables us to learn the optimal transition probability that matches the channel states.

For the end-to-end training of the JCM framework, we leverage the technique of variational learning and derive a matching loss function with physical meanings under the guidance of the information theory.
More specifically, using the principle of Info-Max\cite{linsker1987infomax}, we obtain a variational inference lower bound of the intractable mutual information objective function, which serves as the loss function for the JCM framework. 
By optimizing this loss function, the encoder-modulator can maximize the mutual information, and the decoders can more effectively extract information from the received sequence, resulting in improved decoding performance.

Extensive experiments conducted on image transmission validate three major advantages of our JCM framework. 
First, the proposed JCM framework outperforms the existing separate coding-modulation design via quantization, as well as the hard quantization-based joint design over a wide range of signal-to-noise ratios (SNRs), modulation orders, and transmission rates.
Second, JCM with higher modulation order results in better performances, which provides a convenient way to conduct digital modulation in a semantic communication system, namely by using the highest modulation order that the transmission device could support to achieve the best performance.
Third, JCM can attain a probabilistic shaping of an approximate Gaussian distribution for the AWGN channel without the need for explicit instructions to the NNs regarding the probability distribution of the output signal, which demonstrates its ability to match with the channel conditions.

The rest of the paper is organized as follows. 
In Section II, we introduce our JCM framework and the objective function principles. 
Section III details the training process of the JCM framework, including the derivation of the loss function and the generation of differentiable constellation sequences.
We also present the formulation of the transition probability model in this section. 
In Section IV, we evaluate the performance of our proposed JCM framework through extensive experiments. Finally, we conclude the paper in Section V.

\textit{Notation}: Throughout this paper, we use $H(X)$ to denote the entropy of the variable $X$ and $H(X|Y)$ to denote the conditional entropy of $X$ given $Y$. $I(X;Y)$ refers to the mutual information between $X$ and $Y$.
The statistical expectation of $X$ with probability distribution $p(x)$ is denoted as $\mathbb{E}_{p(x)}\lbrack X \rbrack$.
The Kullback-Leibler (KL) divergence between two probability distributions $p(x)$ and $q(x)$ is denoted as $\mathrm{KL}\lbrack p(x)||q(x) \rbrack$.
We use $\mathcal{N}(\mu, \sigma^2)$ and $\mathcal{CN}(\mu, \sigma^2)$ to respectively denote the real and complex Gaussian distribution with mean $\mu$ and variance $\sigma^2$.
$\mathbf{I}_{k\times k}$ refers to an identity matrix of dimension $k\times k$.
Moreover, $\mathbb{C}^n$ and $\mathbb{R}^n$ represent sets of complex and real vectors of dimension $n$, respectively.

\section{Proposed JCM Framework}

In this section, the proposed JCM framework is presented, including its system model, architecture design and the objective function principles. 

\subsection{System Model}

\begin{figure*}
    \centering
    \includegraphics[width=0.89\textwidth]{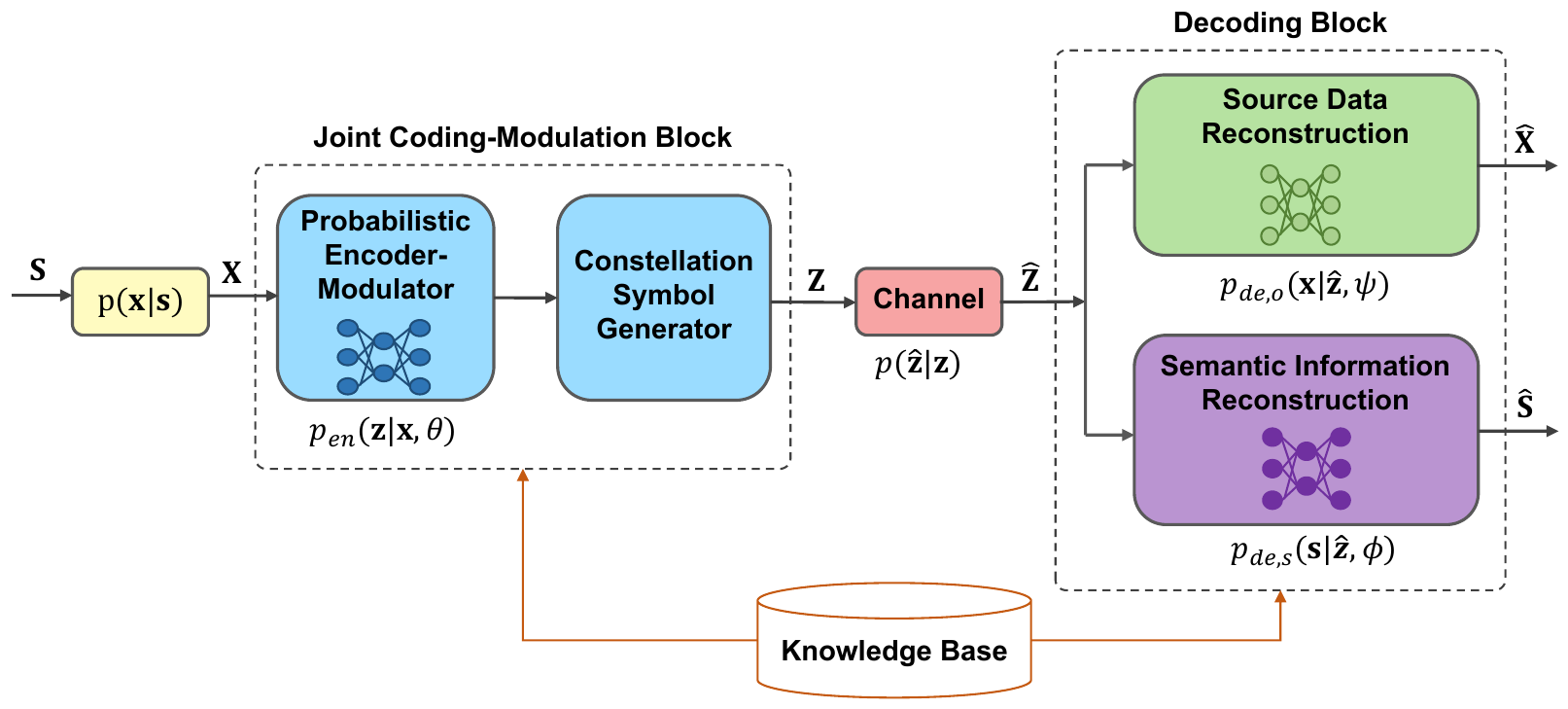}
          \caption{The proposed JCM system model.}
          \label{jcm}
\end{figure*}

The system model of the proposed JCM framework as an end-to-end design for digital semantic communications is shown in Fig.~\ref{jcm}.
The JCM framework is based on the probabilistic encoder-decoder architecture of the VAE\cite{vae} and includes a joint coding-modulation block at the transmitter and two probabilistic decoders at the receiver.
The joint coding-modulation block is responsible for generating the channel input, while two probabilistic decoders recover the source data and the semantic information from the received signal.
The details are stated as follows.

At the transmitter, there is a source data denoted as a variable $\mathbf{X}\in\mathbb{R}^k$ with dimension $k$, associated with its unknown semantic information, denoted as $\mathbf{S}$, to be transmitted through a noisy channel.
The definition of semantic information $\mathbf{S}$ depends on the specific task at the receiver.
We model the source $\mathbf{X}$ as being generated by the semantic information $\mathbf{S}$ from an unknown and complicated transition probability $p(\mathbf{x}|\mathbf{s})$. 
Following the setup of semantic communications for image data \cite{liu2021semantics, zhang2022task}, the receiver needs to recover both the semantic information and the original source data, which is an abstraction of the practical requirements in many real-world scenarios where both humans and machines are involved in the task decision.
We denote the recovery of the source data as $\mathbf{\hat X}$ and the recovery of semantic information as $\mathbf{\hat S}$.
Both the transmitter and the receiver can access to a shared knowledge base, which is essentially a dataset containing different samples of the source data $\mathbf{X}$ and its corresponding semantic information $\mathbf{S}$.

The joint coding-modulation block, which consists of a probabilistic encoder-modulator and a constellation symbol generator, generates channel input $\mathbf{Z}\in\mathbb{C}^n$ from source $\mathbf{X}$, with $n$ being the number of channel uses.
Specifically, parameterized by an NN with parameters $\theta$, the probabilistic encoder-modulator is designed to learn the transition probability $p_{en}(\mathbf{z}|\mathbf{x},\theta)$. 
According to this transition probability, a channel input $\mathbf{z}$ is randomly generated and then sent through a communication channel.
As a side note, we use the Gumbel-Softmax method \cite{jang2016softgumbel} to generate differentiable constellation symbols. 
We will provide a detailed discussion on the formulation of the transition probability model and the generation of the constellation symbol sequence in Section III.

Different from most semantic communication systems, a mandatory requirement of digital modulation is placed in our system. 
We consider an $M$-order digital modulation with constellation symbols denoted as $\cal C$ $=\{c_1, c_2, ..., c_M\}$. 
That is, each element of the channel input $\mathbf{Z}$ takes values from $\cal C$ and thus we call $\mathbf{Z}$ a constellation symbol sequence.
Note that $\mathbf{Z}$ is scaled before actual transmission so that it meets an average transmit power constraint $P$, \emph{i.e.}, $P=\frac{\|\mathbf{Z}\|^2}{n}$.
We model the channel as an additive white Gaussian noise (AWGN) channel. 
As such, the received sequence $\mathbf{\hat Z}$ can be written as $\mathbf{\hat Z}=\mathbf{Z}+\boldsymbol{\varepsilon}$, where $\boldsymbol{\varepsilon}\sim \mathcal{CN}(\mathbf{0}, \sigma ^{2}\mathbf{I})$ is the channel noise, of which each element is drawn independently from a complex Gaussian distribution with zero mean and variance $\sigma^2$.
The channel condition is characterized by the channel SNR, which is defined as $\frac{P}{\sigma^2}$.

The receiver consists of two decoders to respectively reconstruct source data and semantic information. Given the received sequence $\mathbf{\hat{Z}}$, the two NNs respectively estimate the posterior distribution $p_{de,o}(\mathbf{x}|\mathbf{\hat z}, \psi)$ and $p_{de,s}(\mathbf{ s}|\mathbf{\hat z}, \phi)$, where $\psi$ and $\phi$ denote their respective NN parameters.  
Then, at the test stage, we recover the source data and the semantic information via maximum a posteriori (MAP) decoding:
\begin{align}
    \mathbf{\hat{x}} &= \underset{\mathbf{x}}{\mathrm{argmax}}  \  p_{de,o}(\mathbf{x}|\mathbf{\hat z}, \psi),\\
    \mathbf{\hat{s}} &= \underset{\mathbf{s}}{\mathrm{argmax}} \  p_{de,s}(\mathbf{s}|\mathbf{\hat z}, \phi).
\end{align}
Notably, we employ a parallel architecture to recover both $\mathbf{X}$ and $\mathbf{S}$ from $\mathbf{\hat Z}$. 
An alternative design could be a cascading architecture, where $\mathbf{X}$ is first recovered as $\mathbf{\hat X}$, followed by inferring $\mathbf{S}$ from $\mathbf{\hat X}$.
The latter would incur inevitable information loss when inferring $\mathbf{S}$ from $\mathbf{\hat X}$ due to the data processing inequality \cite[Theorem~2.8.1]{cover1999elements}. As such, we choose the former architecture.

\subsection{Objective Function}

We apply Info-Max\cite{linsker1987infomax} as the design principle for the JCM framework. The principle of Info-Max, employed in many works\cite{barber2005infomax,linsker1987infomax}, aims to maximize the mutual information. 
Specifically, our goal for the encoder-modulator NN is that it can preserve as much information about the semantic information and the source data as possible in the received sequence $\mathbf{\hat Z}$.
Therefore, we define our objective function based on mutual information, abbreviated as ``MI-OBJ'', and establish the optimization problem as:
\begin{equation}
  \underset{\theta}{\mathrm{max}} \  I_\theta(\mathbf{S};\mathbf{\hat Z}) + \lambda\cdot I_\theta(\mathbf{\mathbf{X};\hat Z}),
    \label{obj}
\end{equation}
where $\lambda$ is a trade-off hyperparameter to balance the importance of the two mutual information terms, and $\theta$ represents the parameters of the probabilistic encoder-modulator. 
As a remark, the above objection function is different from the objective function of the information bottleneck (IB) principle \cite{tishby2000information} as applied in \cite{zhang2021bottleneck}.
The IB principle aims to find the best coding scheme of the source data that preserves maximal semantic information under the constraint of the code rate.
Therefore, its objective function is the subtraction of two mutual information terms, whereas ours is the sum of two mutual information terms.

To use MI-OBJ as our loss function, however, there are two main challenges. 
First, MI-OBJ is hard to estimate due to the mutual information terms. 
These high dimensional variables with complicated distributions make it hard to estimate their joint distributions and marginal distributions, which consequently makes the direct optimization of MI-OBJ difficult.
Second, MI-OBJ only considers the optimization of the encoder-modulator parameters $\theta$. 
However, the decoders at the receiver require joint training with the encoder-modulator at the transmitter. 
Thus, we need to design a loss function that takes the optimization of the two decoders $p_{de,o}(\mathbf{x}|\mathbf{\hat z}, \psi)$ and $p_{de,s}(\mathbf{s}|\mathbf{\hat z}, \phi)$ into consideration as well.

Therefore, MI-OBJ cannot be directly used as the loss function for our NNs. A solution is hence given in Section III, which gives a lower bound of MI-OBJ for the joint optimization of the encoder-modulator and the decoders, so that the constellation symbol sequence that preserves the most information can be decoded by the optimal decoder. 

\section{Variational Learning of The JCM Framework}

This section focuses on the training process of the JCM framework.
As mentioned, in JCM we need to optimize the probabilities of ${p_{en}(\mathbf{z}|\mathbf{x}, \theta)}$, $p_{de,o}(\mathbf{ x}|\mathbf{\hat z}, \psi)$ and $p_{de,s}(\mathbf{s}|\mathbf{\hat z}, \phi)$, which falls in the field of variational learning. 
Consequently, new loss function and different gradient estimation methods are required for the variational learning of these probabilities.

\subsection{Loss Function Design Based on  Variational Inference Lower Bound}

In this subsection we derive a general loss function for the training of the JCM framework to jointly optimize the the transition probability ${p_{en}(\mathbf{z}|\mathbf{x}, \theta)}$ at the transmitter, and the probabilities $p_{de,o}(\mathbf{ x}|\mathbf{\hat z}, \psi), p_{de,s}(\mathbf{s}|\mathbf{\hat z}, \phi)$ at the receiver. 
In general, we use an iterative strategy for this joint training process.  
Specifically, we first fix the transition probability of the encoder-modulator NN ${p_{en}(\mathbf{z}|\mathbf{x}, \theta)}$, and then update the probabilities at the receiver $p_{de,o}(\mathbf{ x}|\mathbf{\hat z}, \psi), p_{de,s}(\mathbf{s}|\mathbf{\hat z}, \phi)$. 
These probabilities should be updated in a way to approach the true posterior probabilities $p(\mathbf{ x}|\mathbf{\hat z}), p(\mathbf{s}|\mathbf{\hat z})$ under the fixed ${p_{en}(\mathbf{z}|\mathbf{x}, \theta)}$.
According to the performance of the decoder NNs updated in the last step, we then update the encoder-modulator NN at the transmitter.
These two steps together make up a full training epoch. 
We find that updating the NN at the transmitter is relatively easy, as long as we can find the best decoder under this given transition probability ${p_{en}(\mathbf{z}|\mathbf{x}, \theta)}$. 
Therefore, our main problem is to find a way to train the decoding NNs so that the NN parameterized posterior distributions approach the true posterior distributions. Thus, our efforts mainly focus on designing a reasonable and efficient loss function for the NNs at the receiver.

To address this problem, we apply the technique of variational inference. Variational inference is a technique used in machine learning to approximate complex posterior probability distributions in Bayesian models for which exact inference is computationally intractable\cite{blei2017variational}. 
In JCM, we fit an approximate inference model to the true posterior distributions $p(\mathbf{ x}|\mathbf{\hat z}), p(\mathbf{s}|\mathbf{\hat z})$ and then use variational inference to derive a tractable lower bound of MI-OBJ that includes a distance measure between the true posterior distributions and the NN parameterized posterior distributions $p_{de,o}(\mathbf{ x}|\mathbf{\hat z}, \psi)$, $p_{de,s}(\mathbf{s}|\mathbf{\hat z}, \phi)$.
This lower bound is stated in Theorem \ref{th1}.

\begin{theorem}[VILB] \label{th1}
A variational inference lower bound of MI-OBJ is given by \eqref{elbo}, shown at the top of the next page, where $K=H(\mathbf{S})+\lambda\cdot H(\mathbf{X})$ is a constant.
\end{theorem}



To prove Theorem \ref{th1}, first we expand the decoder probabilities $p_{de,s}(\mathbf{s}|\mathbf{\hat z}, \phi)$ and $p_{de,o}(\mathbf{x}|\mathbf{\hat z}, \psi)$ as variational approximations to the true posterior distributions, and merge them into MI-OBJ. In this way, for each mutual information term we get a lower bound, respectively
\begin{align}
\setcounter{equation}{4}
\mathbb{E}_{p(\mathbf{\hat z})}\mathbb{E}_{p(\mathbf{s}|\mathbf{\hat z})} \log p_{de,s}(\mathbf{ s}|\mathbf{\hat z}, \phi) + H(\mathbf{S}),
\end{align}
and 
\begin{equation}
\mathbb{E}_{p(\mathbf{\hat z})} \mathbb{E}_{p(\mathbf{x}|\mathbf{\hat z})} \log p_{de,o}(\mathbf{ x}|\mathbf{\hat z}, \psi) + H(\mathbf{X}).
\end{equation}
Then, we obtain VILB by using the Markov chain to expand $p(\mathbf{\hat{z}})$.
The complete proof of Theorem \ref{th1} can be found in Appendix A.

Theorem \ref{th1} gives an operational lower bound of MI-OBJ, which we define as the general loss function $\mathcal{L}_{gen}(\theta, \phi, \psi)$. Through maximizing the lower bound, the NNs, namely the parameters of $\theta$, $\phi$ and $\psi$, can be trained over the optimization process. 
Specifically, in Theorem \ref{th1}, MI-OBJ is lower bounded by the sum of a constant and an expectation of two terms $\mathbb{E}_{p(\mathbf{s}|\mathbf{\hat z})} \log p_{de,s}(\mathbf{ s}|\mathbf{\hat z}, \phi)$ and $\mathbb{E}_{p(\mathbf{x}|\mathbf{\hat z})} \log p_{de,o}(\mathbf{ x}|\mathbf{\hat z}, \psi)$. 
These two terms serve as a distance measure between the true posterior distribution and the NN parameterized approximate posterior distribution.
As the true posterior distribution is fixed, the maximization of these two terms enables the approximate posterior distributions to keep approaching their respective true posterior distribution, namely the best possible decoder.
When the approximate posterior distribution is equal to the true posterior distribution, the equality in \eqref{elbo} will hold and VILB will reach its maximum.
Therefore, through maximizing VILB, the encoder-modulator can learn to maximize the mutual information, and the decoders can learn to better extract information from the received sequence and therefore improve their decoding performances.

As a remark, VILB is indeed tractable. 
In the training stage when we have the true semantic information and the true source data, we can use the empirical distribution to substitute the true posterior distribution. 
The approximate posterior distributions are obtained through the decoder NNs. 
Consequently, VILB can be computed using these probabilities. 
On a side note, we can use sampling methods such as Monte Carlo sampling\cite{andrieu2003mcmc} to estimate the expectation. 
Therefore, we use VILB as our tractable loss function for the joint optimization of the encoder-modulator and the two decoders.

\subsection{Transition Probability Model of The Encoder-Modulator}

In JCM, the constellation symbol sequence is generated from the transition probability, which is parameterized by the encoder-modulator NN. However, for an $n$-length constellation symbol sequence of an $M$-order modulation, there are a total of $M^n$ categories of probability that need to be learned.
To simplify the learning, we therefore introduce a model of the transition probability in this subsection.

\begin{figure*}[t]
\setcounter{equation}{3}
\begin{align}
I_\theta(\mathbf{S};\mathbf{\hat Z}) + \lambda\cdot I_\theta(\mathbf{\mathbf{X};\hat Z}) &\ge 
\mathbb{E}_{p(\mathbf{\hat z}|\mathbf{z})}\mathbb{E}_{p_{en}(\mathbf{z}|\mathbf{x}, \theta)}\mathbb{E}_{p(\mathbf{s},\mathbf{x})} \lbrace \mathbb{E}_{p(\mathbf{s}|\mathbf{\hat z})} \log p_{de,s}(\mathbf{s}|\mathbf{\hat z}, \phi)
+\lambda\cdot \mathbb{E}_{p(\mathbf{x}|\mathbf{\hat z})} \log p_{de,o}(\mathbf{ x}|\mathbf{\hat z}, \psi)\rbrace + K\nonumber\\ &\overset{\underset{\mathrm{def}}{}}{=}\mathcal{L}_{gen}(\theta, \phi, \psi)
    \label{elbo}
\end{align}
\setcounter{equation}{9}
\begin{equation}
    p_{en}(\mathbf{z}|\mathbf{x}, \theta)=\textstyle\prod_{i=1}^{n} \textstyle\prod_{r=-\frac{\sqrt{M}}{2}}^{\frac{\sqrt{M}}{2}-1}\textstyle\prod_{s=-\frac{\sqrt{M}}{2}}^{\frac{\sqrt{M}}{2}-1}( q_{ir|\mathbf{x},\theta}^{I})^{\mathbb{I}{\left\lbrace z_{Ii}=\frac{2r+1}{\sqrt{M}-1}\right\rbrace}} \cdot (q_{is|\mathbf{x},\theta}^{Q})^{\mathbb{I}{\left\lbrace z_{Qi}=\frac{2s+1}{\sqrt{M}-1}\right\rbrace}}.
    \label{mqam}
\end{equation}
\setcounter{equation}{12}
\begin{align}
\mathcal{L}(\theta, \phi, \psi;\mathbf{X}^N,\mathbf{S}^N)=
    \frac{1}{N}\sum_{n=1}^{N} \Biggl[\Biggl(  \sum_{l=1}^{L} p(\mathbf{s}_n=l|\mathbf{\hat{z}}_\theta^{n})\log p_{de,s}(\mathbf{ s}_n=l|\mathbf{\hat z}_\theta^{n}, \phi) \Biggr)
    -  \lambda \cdot  ||\mathbf{x}_n- f_\psi(\mathbf{\hat z}_\theta^{n})||_2^2  \Biggr]
    \label{empirical obj}
\end{align}
\hrulefill
\end{figure*}

We formulate the general transition probability model $p_{en}(\mathbf{z}|\mathbf{x}, \theta)$ of an $M$-order modulation as such. For an $n$-length constellation symbol sequence $\mathbf{Z}=(Z_1, Z_2, ..., Z_n)$, as a common simplification\cite{nesct}, we model the components as conditionally independent variables, where a component $Z_i$ takes values in the $M$ symbols of the constellation $\cal C$. 
This assumption decreases the number of probability categories to be learnt from $M^n$ to $M\cdot n$, greatly reducing the computation complexity.
The NN outputs the categorical distribution of each symbol $Z_i$, and we denote the probability of $Z_i$ being one of the constellation symbol $c_m\in \cal C$ as:
\begin{align}
\setcounter{equation}{6}
q_{im|\mathbf{x},\theta}=P(Z_i=c_m|\mathbf{x}, \theta).
\label{norm}
\end{align}
Due to the assumed conditional independence, the joint probability mass function (PMF) of the encoder-modulator can be expressed as:
\begin{equation}
    p_{en}(\mathbf{z}|\mathbf{x}, \theta)=\textstyle\prod_{i=1}^{n}p(z_i|\mathbf{x}, \theta)=\textstyle\prod_{i=1}^{n} \textstyle\prod_{m=1}^{M} q_{im|\mathbf{x},\theta}^{\mathbb{I}\left\lbrace z_i=c_m \right\rbrace},
    \label{generalpmf}
\end{equation}
where $\mathbb{I}{\left\lbrace\cdot\right\rbrace}$ denotes the indicator function.



Since BPSK with $\cal{C}$ $=\lbrace+1, -1\rbrace$ and rectangular M-QAM with $\cal{C}$ $=\lbrace \frac{2r+1}{\sqrt{M}-1} + j \cdot\frac{2s+1}{\sqrt{M}-1}\rbrace, r,s=\frac{-\sqrt{M}}{2},-\frac{\sqrt{M}}{2}+1,...,\frac{\sqrt{M}}{2}-1$ (we use $j$ to denote the imaginary unit) are two commonly used modulation schemes, we give the transition probability of BPSK and rectangular M-QAM as two special cases in the following propositions.
\begin{proposition}[Transition probability of BPSK]
\label{bpskCorollary}
Let $q_{i|\mathbf{x},\theta}$ denote the NN output probability of $Z_i=1$. Then, the PMF of the probabilistic encoder-modulator of BPSK is:
\begin{equation}
    p_{en}(\mathbf{z}|\mathbf{x}, \theta)=\textstyle\prod_{i=1}^{n} (q_{i|\mathbf{x},\theta}) ^{\frac{1}{2}(1+z_i)}(1- q_{i|\mathbf{x},\theta}) ^{\frac{1}{2}(1-z_i)}.  \label{bpsk}
\end{equation}
\end{proposition}

\begin{proposition}[Transition probability of M-QAM]
\label{qamCorollary}
Let $q_{ir|\mathbf{x},\theta}^{I}$ and $q_{is|\mathbf{x},\theta}^{Q}$ respectively denote the probability of I channel amplitude $Z_{Ii}=\frac{2r+1}{\sqrt{M}-1}$ and Q channel amplitude $Z_{Qi}=\frac{2s+1}{\sqrt{M}-1}$, where $Z_i = Z_{Ii} + j\cdot Z_{Qi}$. The PMF of the probabilistic encoder-modulator of rectangular M-QAM can be written as \eqref{mqam} at the top of this page.

\end{proposition}

For BPSK, the transition probability degrades to a Bernoulli distribution. For M-QAM, we consider the I channel and the Q channel to be conditionally independent as well. Therefore, for each $Z_i$, we only need to learn two probability distributions with $\sqrt{M}$ categories instead of one probability distribution with $M$ categories, which considerably decreases the computation complexity. The proof for Proposition \ref{bpskCorollary} and \ref{qamCorollary} can be found in Appendix B and C, respectivley.
Need to mention that, our loss function can be generalized to arbitrary constellation maps, including non-uniform constellation maps.

\subsection{Loss Function for Image Semantic Communications}

VILB derived in Section III-A is a general lower bound for MI-OBJ. In this subsection, based on this general lower bound, we derive a specific loss function for image semantic communications. We first introduce some widely recognized assumptions of the image source data. Then, we derive the loss function by applying these assumptions to VILB and replacing the probability function with the empirical distribution.

The assumptions we take are as follows. As pointed out in \cite{nesct}, for image source, the source data $\mathbf{X}$ and its reconstruction $\mathbf{\hat{X}}$ can be modeled as multivariate factorized Gaussian variables with isotropic covariance. Thus, we can assume that 
\begin{align}
\setcounter{equation}{10}
  &   p(\mathbf{x|\hat z})= \mathcal{N}(\boldsymbol{\mu},\sigma_1^2\mathbf{I}_{k\times k}) \label{true posterior},\\
&p_{de,o}(\mathbf{\hat x|\hat z},\psi)= \mathcal{N}(f_{\psi}(\mathbf{\hat z}),\sigma_2^2\mathbf{I}_{k\times k}), \label{variational}
\end{align}
where $\boldsymbol{\mu}$ is the true pixel value of the source data (usually normalized between $0$ and $1$); $\sigma_1$ and $\sigma_2$ are precision parameters treated as constants\cite{nesct}; $k$ represents the dimension of the image data $\mathbf{X}$. The decoder NN outputs the mean $f_{\psi}(\mathbf{\hat z})$ of the Gaussian distributed approximate posterior distribution.
As for the semantic information $\mathbf{S}$, we set the image classification labels as the semantic information, which makes $\mathbf{S}$ $\in\{1,...,L\}$  a discrete variable where $L$ represents the number of categories.
Bringing these additional assumptions into VILB, we further derive a loss function for image semantic communications as a corollary of Theorem \ref{th1}, where we denote $(\mathbf{X}^N,\mathbf{S}^N)=\left \lbrace(\mathbf{x}_n, \mathbf{s}_n)\right\rbrace_{n=1}^{N}$ as the batch of training data with $N$ being the number of training samples in the dataset.
\begin{corollary}[Loss function for image semantic communications]
\label{corollaryImage}
Consider image semantic communications with $L$-class classification labels as the semantic information. Let $\mathbf{\hat z}_\theta^{n}$ be the received sequence for the $n$th training sample. Then, the loss function for image semantic communications can be written as \eqref{empirical obj} at the top of this page.
\end{corollary}


The loss function \eqref{empirical obj} is a weighted sum of two terms. The first term is the cross entropy (CE), a commonly used loss function for image classification, which measures the distance between the true posterior distribution and the NN parameterized approximate posterior for the semantic information. The second term is the mean square error (MSE) between the true source image and the recovered source image. By minimizing the distance between the NN parameterized approximate posterior distributions and the true posterior distributions using CE and MSE, the decoders can better infer the semantic information and the source data from the received sequence $\mathbf{\hat z}_\theta^{n}$ respectively.

The derivation of Corrollary \ref{corollaryImage} consists of two steps. First, we apply the assumptions of the semantic information and the source data into VILB. Since in this scenario the semantic information is the classification label and thus a discrete variable, the corresponding term in VILB becomes the commonly used CE loss. As for the source data, the term $\mathbb{E}_{p(\mathbf{x}|\mathbf{\hat z})} \log p_{de,o}(\mathbf{ x}|\mathbf{\hat z}, \psi)$ between two Gaussian distributions can be reduced to the MSE.
Then, we use law of large numbers to replace the probability distributions with the empirical distributions, and get the loss function in \eqref{empirical obj}. Through maximizing this loss function, or equivalently minimizing its negative, the NN parameters of $\theta$, $\phi$ and $\psi$ can be optimized in the training process. The complete proof of Corrollary \ref{corollaryImage} can be found in Appendix D.

\subsection{Differentiable Constellation Symbol Generation}

Using the SGD algorithm to optimize the NN parameters $\theta$, $\phi$ and $\psi$ requires the estimation of their gradients for the loss function. 
For the decoder parameters $\phi,\psi$, by interchanging the differentiation and expectation in the general loss function in \eqref{elbo}, we can directly use the Monte Carlo estimator to estimate their gradients\cite{mohamed2020monte}. However, estimating the gradients of the encoder-modulator parameters $\theta$ causes problems,
since $\theta$ defines the distribution that is integrated in the expectation, which prevents the direct application of the Monte Carlo estimator.
To further explain this point, we rewrite the general loss function as $\mathbb{E}_{ p(\mathbf{x}, \mathbf{s})}\mathbb{E}_{ p_{en}(\mathbf{z}|\mathbf{x},\theta)}\left\lbrack h_{\phi, \psi}(\mathbf{z})\right\rbrack$, where $h_{\phi, \psi}(\mathbf{z})$ denotes the function of $\mathbf{z}$ that is not dependant on $\theta$. Therefore, the gradients with respect to $\theta$ can be written as $\nabla_\theta\mathbb{E}_{ p(\mathbf{x}, \mathbf{s})}\mathbb{E}_{ p_{en}(\mathbf{z}|\mathbf{x},\theta)}\left\lbrack h_{\phi, \psi}(\mathbf{z})\right\rbrack$.
The general principle is to rewrite such gradients in a form that allows for the Monte Carlo estimator, so that they can be easily and effectively computed \cite{mohamed2020monte}.
The commonly used score function estimator solves this problem by using the identity $\nabla_\theta p_{en}(\mathbf{z}|\mathbf{x},\theta)=p_{en}(\mathbf{z}|\mathbf{x},\theta) \nabla_\theta \log p_{en}(\mathbf{z}|\mathbf{x},\theta)$. Apply this identity, and we can get
\begin{align}
\setcounter{equation}{13}
    &\nabla_\theta\mathbb{E}_{ p(\mathbf{x}, \mathbf{s})}\mathbb{E}_{ p_{en}(\mathbf{z}|\mathbf{x},\theta)}\left\lbrack h_{\phi, \psi}(\mathbf{z})\right\rbrack \nonumber\\ &= \mathbb{E}_{ p(\mathbf{x}, \mathbf{s})}\mathbb{E}_{ p_{en}(\mathbf{z}|\mathbf{x},\theta)}\left\lbrack h_{\phi, \psi}(\mathbf{z})\nabla_\theta \log p_{en}(\mathbf{z}|\mathbf{x},\theta)  \right\rbrack.
\end{align}
Then, this expectation can be stochastically approximated by the Monte Carlo estimator. However, this method suffers from high variance as well as slow convergence \cite{jang2016softgumbel}, and is often used with other variance reduction techniques\cite{mnih2014varReduction2,gregor2014varRedunction1}, which increases the computation complexity of the training.


The reparameterization trick popularized in VAE\cite{vae} provides a simple yet effective solution to the gradient estimation of $\theta$.
In VAE, to update the parameters of a Gaussian variable whose distribution is integrated in the expectation, the authors reparameterize the Gaussian variable as a deterministic function of an independent random variable and the parameters to be updated. In this way, an unbiased gradient estimator for the parameters can be obtained and easily optimized, enabling the backpropagation through one sample of the Gaussian variable.
Empirical results show that the reparameterization trick effectively reduces the estimation variance\cite{vae}, and allows for the efficient training of the NN.

We adopt the reparameterization trick to generate the constellation symbols that involves a method called the Gumbel-Softmax method\cite{jang2016softgumbel}.
Unlike Gaussian variables in VAE, constellation symbols in our considered digital semantic communications are discrete variables with a categorical distribution.
The Gumbel-Softmax method allows us to generate constellation symbol sequences during the forward propagation and estimate the gradients with respect to the encoder-modulator NN parameters during backpropagation. 
Specifically, during the forward propagation, we use the Gumbel-Max sampling method\cite{gumbel954gumbelproof} to reparameterize the sampling process. 
During backpropagation,  we replace the non-differentiable $\mathrm{argmax}$ function used in Gumbel-Max with the differentiable $\mathrm{softmax}$ function for gradient estimation.
The general Gumbel-Max sampling method, which reparameterizes a discrete variable using an independent Gumbel variable, can be summarized as follows:
Consider a \text{one-hot} variable $\mathbf{t}=(t_1,...,t_a)\in\left\lbrace0, 1\right\rbrace^{a}$ with $\sum_{i=1}^a t_i=1$ where the element 1 indicates different categories. Its categorical probability distribution is denoted by $\mathbf{\pi}=(\pi_1, ..., \pi_a)$. Let $\tau_i$ denote a random variable with Gumbel distribution $\tau_{i}\sim \mathrm{Gumbel}(0,1)$. Variable $\mathbf{t}$ can be sampled from $\mathbf{\pi}$ via:
\begin{equation}
    \mathbf{t} = \text{one-hot}\left(\underset{i\in\left\lbrace1,...,a\right\rbrace}{\mathrm{argmax}}\left\lbrack \tau_i + \log \pi_i\right\rbrack\right),
\end{equation}
where $\text{one-hot}(x)$ converts the index $x$ into its one-hot form, in which the $x$-th element is equal to 1, and the rest are equal to 0, respectively standing for ``on'' and ``off''. 

We use the Gumbel-Max sampling method to generate the constellation symbol sequence in the forward propagation.
Specifically, we sample from the transition probability $p_{en}(\mathbf{z}|\mathbf{x},\theta)$ to get a constellation symbol sequence, so that its empirical distribution of matches its probability distribution. 
In an $M$-order modulation, the sampling process of $z_i$ from the transition probability $p_{en}(\mathbf{z}|\mathbf{x},\theta)$ can be expressed as:
    \begin{equation}
        z_i = y^{i}_\theta(\mathbf{x}, \boldsymbol{\tau})=\mathbf{c}^{T}\cdot \text{one-hot}\left(\underset{m\in\left\lbrace 1,...,M \right\rbrace}{\mathrm{argmax}}\left\lbrack \tau_{im} + \log q_{im|\mathbf{x},\theta}\right\rbrack\right),
            \label{ourSampling}
    \end{equation}
where $\mathbf{c}=\left(c_1,...,c_M\right)$ is the vector composed of all the constellation symbols, $\tau_{im}$ is a random variable with Gumbel distribution and $q_{im|\mathbf{x},\theta}$ is the probability as in \eqref{norm}. 
The length of the one-hot vector corresponds to the total number of the constellation symbols $M$.
When the $x$-th element is equal to 1, it means that the $x$-th constellation symbol has been sampled.
The dot product of $\mathbf{c}^{T}$ and the 
\text{one-hot} vector samples the constellation symbol for $z_i$.


In the Gumbel-Softmax reparameterization trick, to estimate the gradients with respect to $\theta$ in backpropagation, the non-differentiable $\mathrm{argmax}$ function is approximated by the differentiable $\mathrm{softmax}$ function.
Hence, the \text{one-hot} vector obtained by the Gumbel-Max sampling method is substituted by a vector $\mathbf{v}_i=(v_{i1},...,v_{iM})$ whose $m$th entry is:
\begin{equation}
    v_{im}=\frac{\exp((\log q_{im|\mathbf{x},\theta} + \tau_{im})/\rho)}{\sum_{k=1}^{M}\exp((\log q_{ik|\mathbf{x},\theta} + \tau_{ik})/\rho)},
\end{equation}
where $\rho$ represents the temperature hyperparameter to control the hardness of the approximation.
Therefore, different from the forward propagation, in backpropagation we have $\mathbf{z}=\Tilde{\mathbf{y}}_\theta(\mathbf{x},\boldsymbol{\tau})$, whose $i$th entry is determined by
\begin{equation}
    z_i=\Tilde{y}_{\theta}^{i}(\mathbf{x},\boldsymbol{\tau})=\mathbf{c}^{T}\cdot\mathbf{v}_i.
\end{equation}
In this way, the loss function is differentiable everywhere, and the gradients with respect to the parameters $\theta$ can be easily obtained through samples of $\mathbf{Z}$ in backpropagation. The gradient estimator of $\theta$ can thus be obtained as
\begin{align}
    &\nabla_\theta\mathbb{E}_{ p(\mathbf{x}, \mathbf{s})}\mathbb{E}_{ p_{en}(\mathbf{z}|\mathbf{x},\theta)}\left\lbrack h_{\phi, \psi}(\mathbf{z})\right\rbrack\nonumber\\
    &=\nabla_\theta\mathbb{E}_{ p(\mathbf{x}, \mathbf{s})}\mathbb{E}_{p(\boldsymbol{\tau})}\left\lbrack h_{\phi, \psi}(\Tilde{\mathbf{y}}_\theta(\mathbf{x},\boldsymbol{\tau}))\right\rbrack\nonumber\\
    &=\mathbb{E}_{ p(\mathbf{x}, \mathbf{s})}\mathbb{E}_{p(\boldsymbol{\tau})}\left\lbrack \nabla_\theta h_{\phi, \psi}(\Tilde{\mathbf{y}}_\theta(\mathbf{x},\boldsymbol{\tau}))\right\rbrack.
    \label{graEstimation}
\end{align}

\section{Experiment Results}

In this section, we provide comprehensive experiments to validate the advantage of the proposed JCM framework at various channel states, transmission rates and modulation orders.
The experiments are performed on an Intel Xeon Silver 4214R CPU, and a 24 GB Nvidia GeForce RTX 3090 Ti graphics card with Pytorch powered with CUDA 11.4.


\subsection{Experiment Settings}
\subsubsection{Datasets}
If not mentioned otherwise, the experiments are conducted on the CIFAR10 dataset\cite{krizhevsky2009cifar}, where the classification label is set as the semantic information. The CIFAR10 dataset consists of 60,000 $32\times 32$ color images in 10 classes, among which 50,000 images are used as training data and 10,000 images are used as test data.

\subsubsection{Neural Network Architecture And Hyper-parameters}

Table \ref{architecture} describes the detailed NN architecture of our method. The probabilistic encoder-modulator employs Resnet \cite{resnet} as the backbone, and its output is then sent into an MLP with an output dimension of $2\sqrt{M}\times n$, the number of probability categories required for an $n$-length constellation symbol sequence with $M$-order modulation. 
We adopt Spinal-net\cite{spinal} for the semantic information reconstruction.
For the source data reconstruction, we use Resnet combined with the depth-to-space operation to perform the upsampling.
We use a batchsize of 32 samples and employ the Adam optimizer for the training. 
We use a cosine annealing schedule with an initial learning rate as $5\times10^{-4}$, which then gradually decreases to $10^{-6}$ in the duration of 300 epochs according to $lr(t)=10^{-6}+\frac{1}{2}(5\times10^{-4}-10^{-6})(1+\cos{\frac{t}{300}\pi})$, where $t$ represents the current epoch number.
The temperature hyperparamter $\rho$ is set to be 1.5.
Moreover, we select the value of the tradeoff hyperparameter $\lambda$ in the loss function \eqref{empirical obj} to strike a balance where classification accuracy is maintained at a consistently high level while maximizing the image recovery performance.
The value of $\lambda$ varies with the system parameters such as the channel SNR, the modulation order and the transmission rate. 
The specific values of $\lambda$ used in our simulation are presented in Table \ref{values of lambda}.

\begin{table}[t]
    \centering
    \caption{Neural network architecture of the proposed method, where $n$ is the number of channel uses and $M$ is the order of modulation.}
    \label{architecture}
    \begin{tabular}{|c|c|c|}
    \hline
            & \textbf{Layer} & \textbf{\makecell[c]{Output\\ Dimension}} \\
         \hline
        \multirow{5}{*}{\textbf{\makecell[c]{Probabilistic\\ Encoder\\-Modulator}}}  & Conv + BatchNorm + ReLU & 64$\times$32$\times$32 \\
        \cline{2-3}
        \multirow{4}{*}{} & Resnet Block $\times$ 4 & \makecell[c]{4$\times$4$\times n$ (BPSK)\\ 4$\times$4$\times2n$ (M-QAM)} \\
        \cline{2-3}
        \multirow{4}{*}{} & Flattening + MLP & \makecell[c]{$2\times n$ (BPSK)\\$2\sqrt{M}\times n$ (M-QAM)} \\
        \cline{2-3}
        \hline
        \textbf{\makecell[c]{Semantic\\Information\\ Reconstruction}} & Spinalnet Block $\times$ 4 + MLP & 10 \\
        \hline
        \multirow{5}{*}{\textbf{\makecell[c]{Source\\Data\\ Reconstruction}}} & Conv + ReLU & 256$\times$4$\times$4 \\
        \cline{2-3}
        \multirow{5}{*}{} & Resnet Block $\times$ 2 & 256$\times$4$\times$4 \\
        \cline{2-3}
        \multirow{5}{*}{} & Reshape + Conv + ReLU & 128$\times$16$\times$16 \\
        \cline{2-3}
        \multirow{5}{*}{} & Resnet Block & 128$\times$16$\times$16 \\
        \cline{2-3}
        \multirow{5}{*}{} & Reshape + Conv + ReLU & 3$\times$32$\times$32 \\
        \hline
    \end{tabular}
\end{table}

\begin{table}[t]
  \begin{center}
    \caption{The specific values of $\lambda$ for varying channel SNRs and the modulation methods, with channel use $n=128$.}
    \begin{tabular}{|c|c|c|c|c|c|c|c|c|c|c|c}
    \hline
      \diagbox{Order}{$\lambda$}{SNR (dB)} & 18 & 12 & 6 & 0 & -6 & -12 & -18\\
      \hline
      BPSK & 70 & 70 & 70 & 30 & 20 & 2 & 0.5 \\
      \hline
      4, 16, 64QAM & 270 & 250 & 250 & 30 & 20 & 2 & 0.5 \\
      \hline
    \end{tabular}
    \label{values of lambda}
  \end{center} 
\end{table}

\subsubsection{Benchmarks}
We compare the performance of our proposed JCM framework with four benchmarks: the conventional NN-based analog modulation (abbreviated as ``Analog'')\cite{gunduz2019jscc1}, and three quantization-based digital semantic coding-modulation methods respectively via the end-to-end hard-to-soft quantizer\cite{unconstrained}, the non-learning-based uniform quantization (abbreviated as ``Uniform'')\cite{xie2020benchamrk1} and the learning-based quantization (abbreviated as ``NN'')\cite{yeli2022harq}. Two performance criteria are used, namely peak-signal-to-noise ratio (PSNR) for the image recovery task and accuracy for the image classification task. The details of the baseline methods are listed as follows.

\begin{figure*}[t]
    \centering
        \begin{subfigure}[b]{0.99\textwidth}
          \centering
          \includegraphics[width=0.95\textwidth]{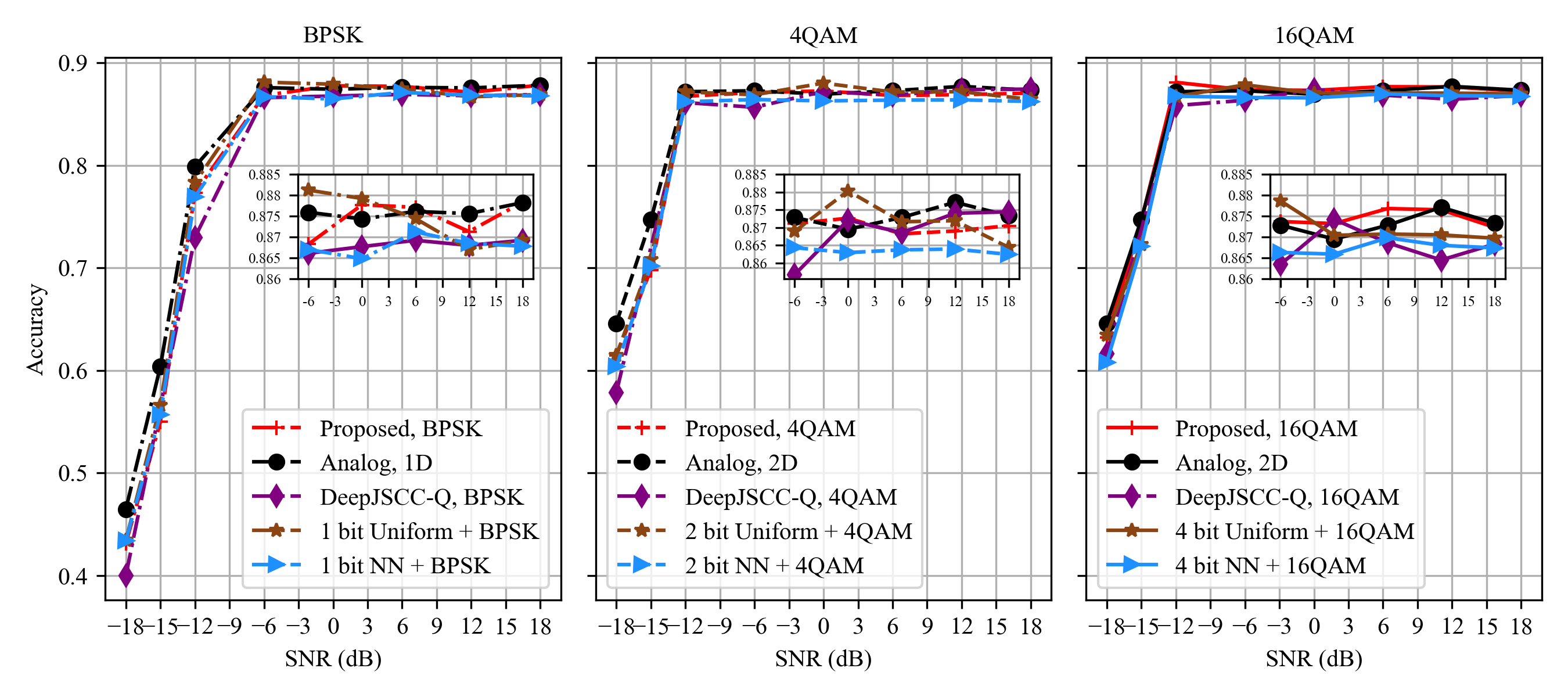}
          \caption{Classification accuracy \emph{vs.} SNR.}
          \label{snr_mod_acc}
        \end{subfigure}
        \begin{subfigure}[b]{0.99\textwidth}
          \centering
          \includegraphics[width=0.95\textwidth]{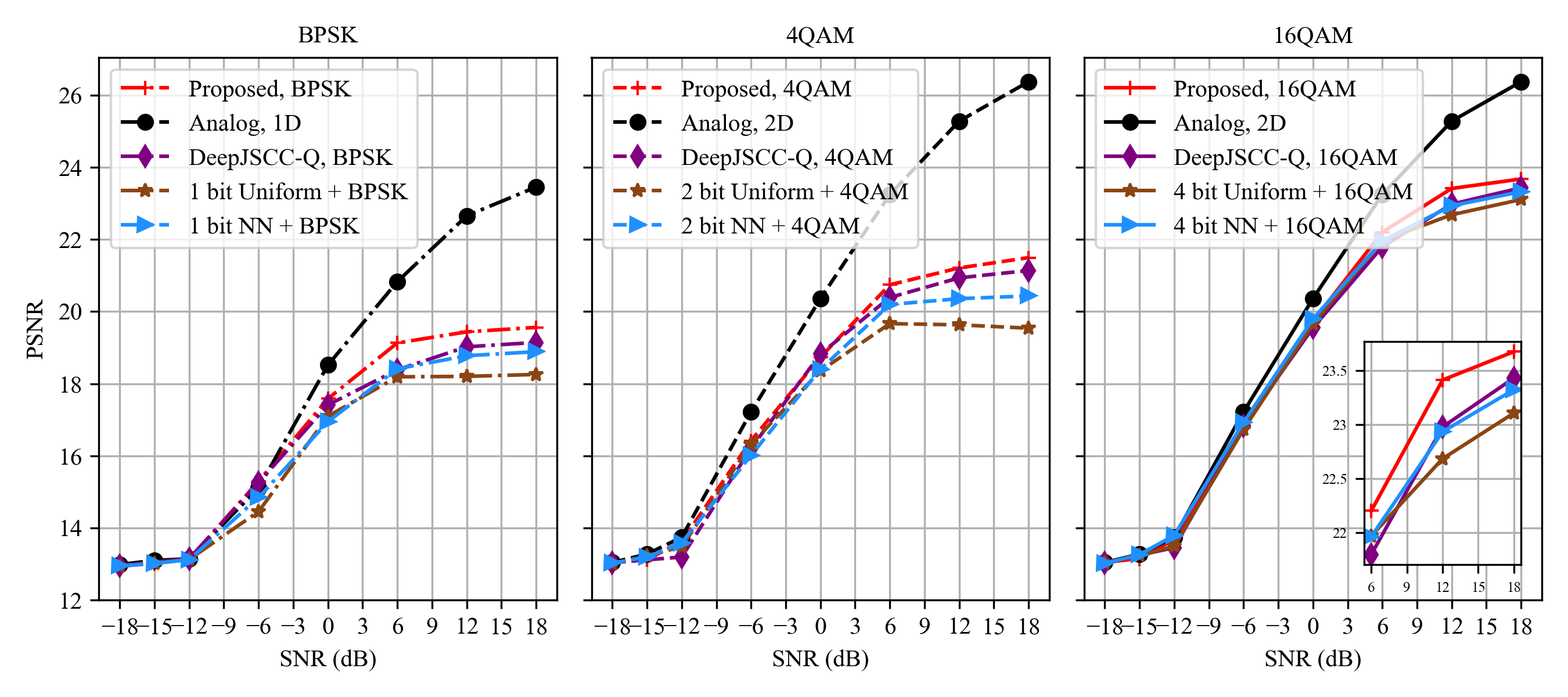}
          \caption{PSNR \emph{vs.} SNR.}
          \label{snr_mod_psnr}
        \end{subfigure}
    \caption{Performances of the classification accuracy and the image recovery with varying channel SNRs and three different modulation schemes: BPSK, 4QAM and 16QAM. The number of channel uses is set at 128.}
    \label{snr_mod}
\end{figure*}

\begin{figure}[t]
    \centering
      \includegraphics[width=0.49\textwidth]{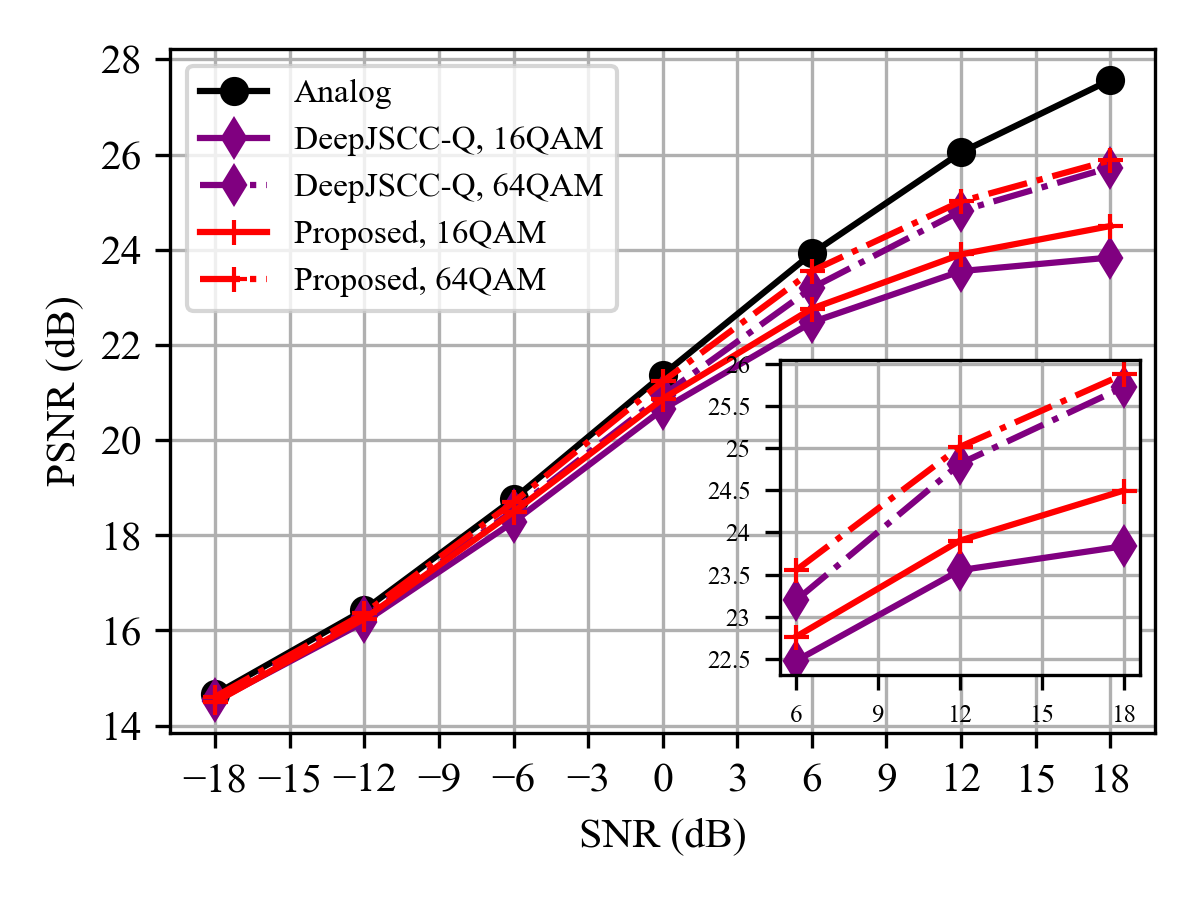}
    \caption{PSNR \textit{vs.} SNR on Tiny Imagenet with 1024 channel uses. The modulation schemes are set as 16QAM and 64QAM.}
    \label{tiny}
\end{figure}

\begin{figure*}[t]
    \centering
        \begin{subfigure}[b]{0.48\textwidth}
          \centering
          \includegraphics[width=0.98\textwidth]{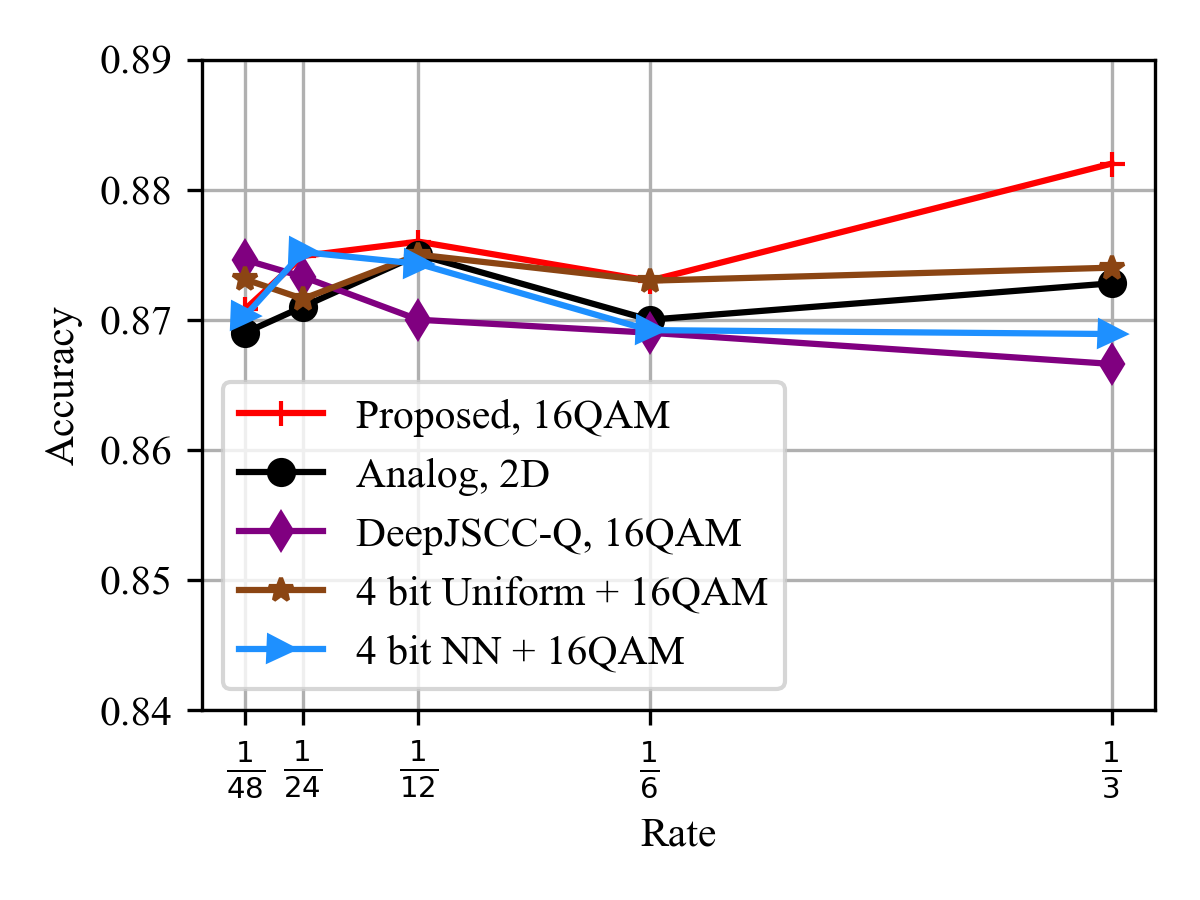}
          \caption{Classification accuracy \emph{vs.} rate.}
          \label{acc qam}
        \end{subfigure}
        \begin{subfigure}[b]{0.48\textwidth}
          \centering
          \includegraphics[width=0.98\textwidth]{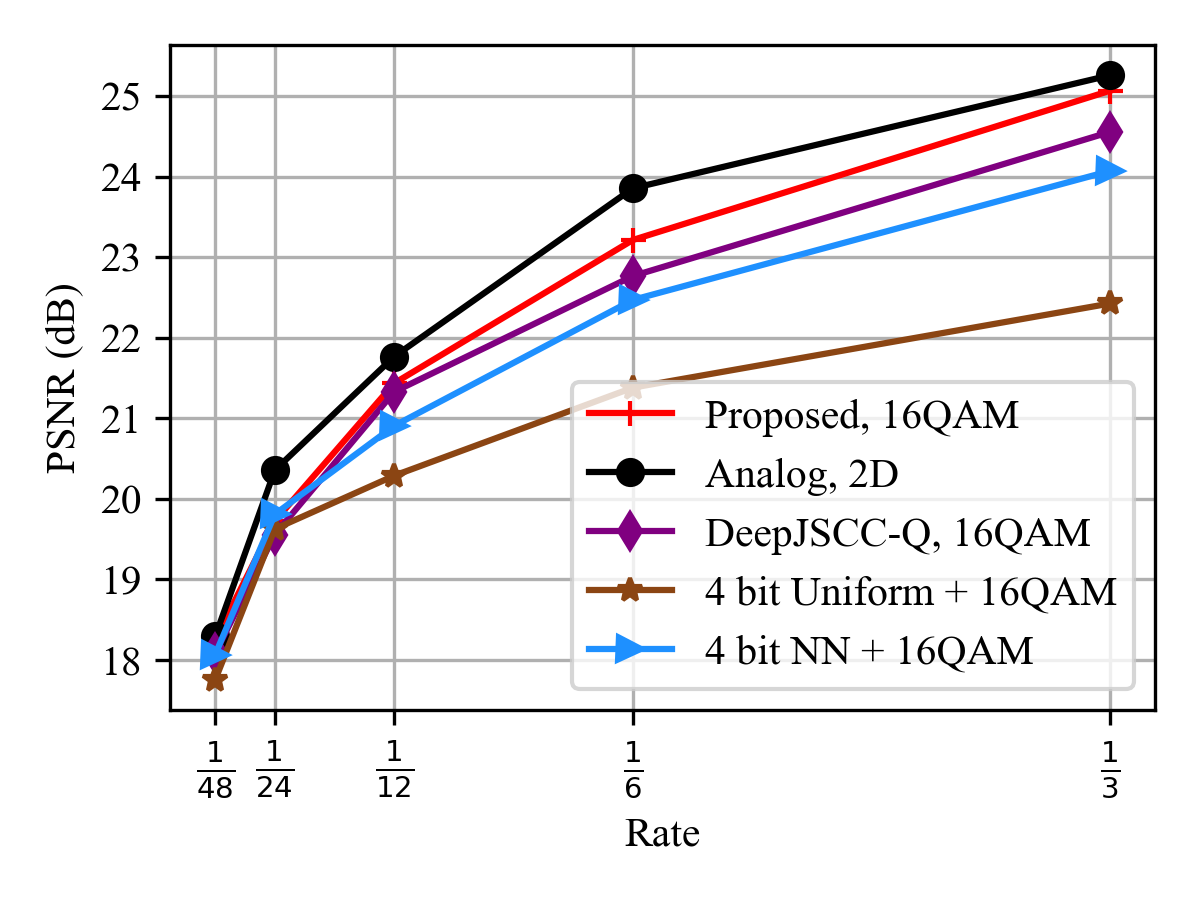}
          \caption{PSNR \emph{vs.} rate.}
          \label{psnr qam}
        \end{subfigure}
    \caption{Performance comparison with varying transmission rates in 16QAM. SNR is set at 0 dB.}
    \label{rate}
\end{figure*}

\begin{figure}[t]
\centering
\includegraphics[width=0.9\linewidth]{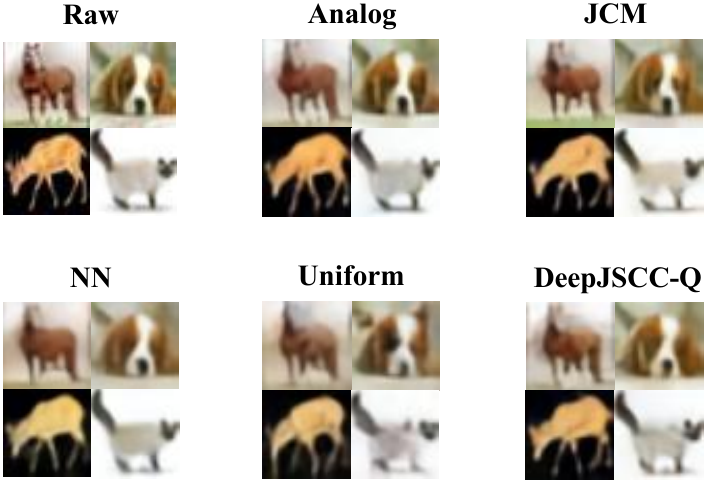}
\caption{Examples of recovered images at SNR $=$ 0 dB with 1024 channel uses. The modulation scheme is set as 16QAM.} 
\label{imageRecovery}
\end{figure}

\begin{itemize}
\item \emph{``Analog'' method}: 
In this method, the output of the NN-based semantic encoder is directly sent into the channel without digitalization. Therefore, this method is called the analog modulation method.
The output dimension of the analog method matches that of the digital methods. In BPSK modulations, the output of the analog method is a one dimensional (1D) sequence with the same channel use as other methods, while in QAM modulations, the output of the analog method is two dimensional (2D).
In terms of the NN architecture, the NN-based semantic encoder also adopts Resnet as the backbone. The two decoders use the same architecture as in the proposed method as well.
As a remark, with unconstrained channel input, this method serves as the performance upperbound for all digital methods\cite{unconstrained}.
\item \emph{DeepJSCC-Q method}: This method jointly trains coding and modulation through the soft-to-hard quantizer designed in \cite{unconstrained}.
In this method, a hard-decision quantization is applied in the forward pass to map the real-valued output of the semantic encoder to the nearest symbol in the constellation.
In the backward pass, this non-differentiable operation is approximated by a Softmax weighted sum.
The whole system is trained end-to-end using the loss function proposed in \cite{unconstrained} while adding an additional term of cross entropy for classification.
\item \emph{``Uniform'' method}: 
In this method, the output of the NN-based semantic encoder is first uniformly quantized into a discrete number sequence, then modulated and sent into channel. Consequently, this method is termed as ``Uniform'' method. 
The NN-based semantic encoder and decoders are the same as those in the analog method.
Additionally, we match the quantization order with the modulation order. For example, we use 1 bit quantization in BPSK, 2 bit quantization in 4QAM, and so on.
\item \emph{``NN'' method}: In this method, the system model is the same as the uniform method, except that the uniform quantizer and dequantizer are replaced by their NN counterparts. Therefore, we call this method ``NN'' method. The NN quantizer and dequantizer utilize a one-layer MLP each, and they are trained to minimize the quantization error, using MSE as the loss function. 
Also, same as the uniform method, the quantization order is matched with the modulation order.
\end{itemize}


\subsection{Performances against Varying SNRs}


In this subsection, the performances of our method and the benchmark methods are evaluated at different SNRs and different modulation schemes. 
If not specified otherwise, the SNR ranges from -18 dB to 18 dB with $n=$ 128 channel uses and the modulation schemes are BPSK, rectangular 4QAM and rectangular 16QAM. 

Fig.~\ref{snr_mod}(a) and Fig.~\ref{snr_mod}(b) respectively plot the classification accuracy and image recovery performances versus SNR in different modulation schemes.
We can observe that all methods achieve nearly identical classification accuracy performances across various SNRs and modulation schemes.
This is due to the relatively simple nature of 10-class classification, which makes the recovery of the semantic information easier compared to that of the source data.
As for PSNR performances, first, we can see that in the low SNR region of SNR $\le-6$ dB, the proposed method has a close performance with the analog method for all three modulation schemes.
The performance gap widens as SNR increases.
Nevertheless, it can be observed that the performance gap between the proposed method and the upperbound analog method decreases as the modulation order increases.
For example, at SNR $=$ 12 dB, the analog method outperforms JCM with 4QAM by 3.2 dB, and outperforms JCM with 16QAM only by 1.9 dB. 
We will study the relationship between the modulation order and the performances more closely in Section IV-D.

Moreover, the proposed method surpasses all the quantization-based digital benchmark methods in all modulation schemes especially in high SNR region.
For BPSK modulation with SNR = 18 dB, the JCM method outperforms the DeepJSCC-Q method, the NN method, and the Uniform method by 0.4 dB, 0.7 dB, and 1.3 dB respectively at PSNR performance.
For 16QAM at SNR = 18 dB, the JCM can provide up to 0.3 dB, 0.4 dB, and 0.6 dB gain, respectively over the three digital benchmarks. 
Note that the advantage of the proposed JCM over DeepJSCC-Q, which also jointly optimizes coding and modulation, stems from our probability-based soft-decision modulation approach.
The JCM method directly learns the probability distribution from the source to the constellation symbols, then samples the constellation sequence for transmission, thereby avoiding the information loss due to hard decisions. 
This advantage can also be understood as the advantage of Variational Autoencoders (VAE) over Autoencders (AE).
By explicitly modeling probabilistic distributions and introducing a level of uncertainty into NNs, VAEs prove more advantageous for tasks involving uncertainty, such as handling channel noise.


For a more comprehensive comparison of our JCM method and the DeepJSCC-Q method, we further evaluate their performance on Tiny Imagenet\cite{le2015tiny}, a larger-scale image dataset consisting of 100,000 64$\times$64 colored images. 
On a side note, on Tiny Imagenet, we exclusively focus on the PSNR performance comparison shown in Fig.~\ref{tiny}, since as shown on the CIFAR10 dataset, all methods exhibit comparable classification accuracy, and their performance disparities are mainly reflected through PSNR.
From Fig.~\ref{tiny}, we can see that the proposed method outperforms the DeepJSCC-Q method with both 16QAM and 64QAM.
With 16QAM, for example, the proposed method outperforms the DeepJSCC-Q method by 0.66 dB at SNR = 18 dB.
With 64QAM, the proposed method outperforms the DeepJSCC-Q method by 0.15 dB at SNR = 18 dB.

\begin{figure*}[t]
    \centering
        \begin{subfigure}[b]{0.48\textwidth}
          \centering
          \includegraphics[width=1\textwidth]{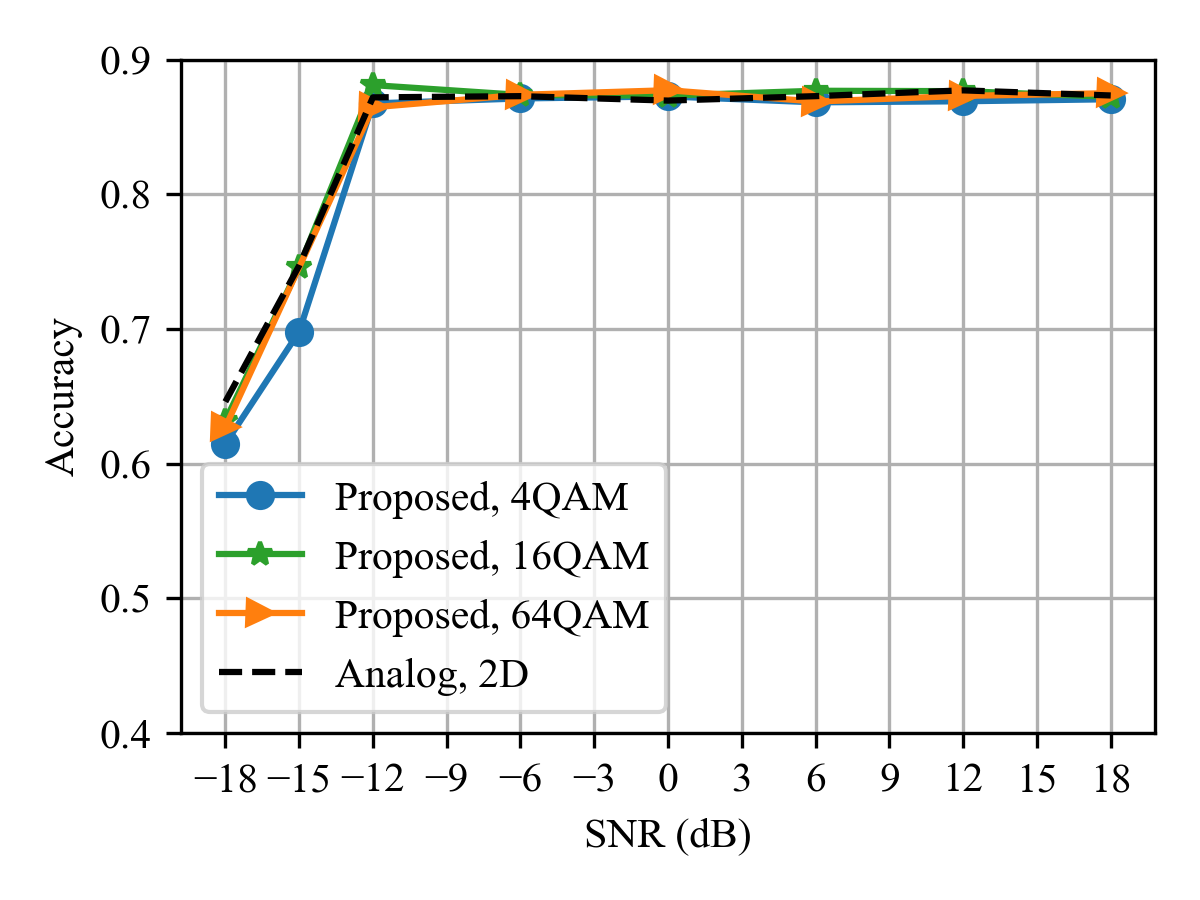}
          \caption{Classification accuracy \textit{vs.} SNR.}
        \end{subfigure}
        \begin{subfigure}[b]{0.48\textwidth}
          \centering
          \includegraphics[width=1\textwidth]{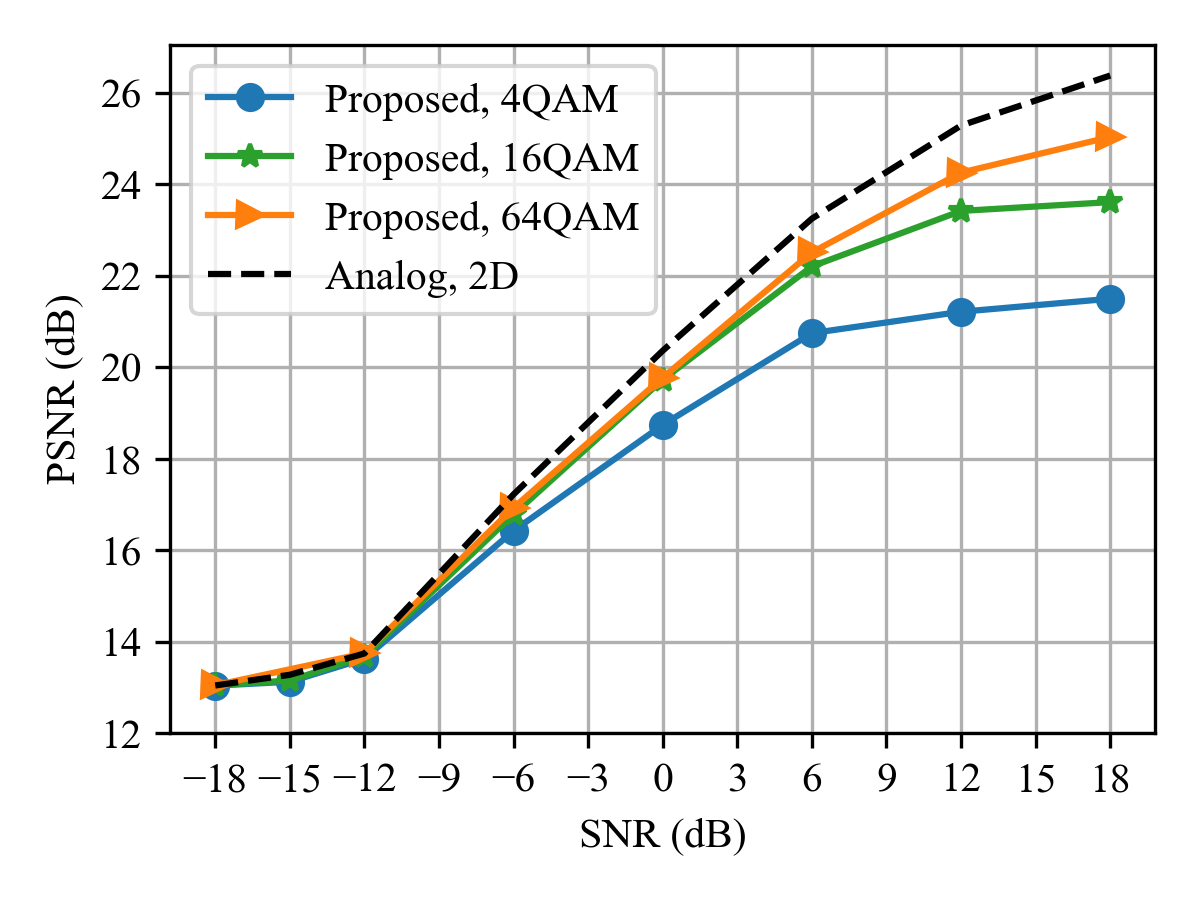}
          \caption{PSNR \textit{vs.} SNR.}
        \end{subfigure}
    \caption{Performance comparison of 4QAM, 16QAM and 64QAM of the proposed method.}
    \label{64qam}
\end{figure*}

In conclusion, the proposed method has a performance advantage over other digital benchmark methods, especially in high SNR region, and is upperbounded by the analog method.
The performance gap between the proposed method and the analog method decreases with increasing modulation order.

\subsection{Performances against Varying Transmission Rates}


In this subsection, we study the relationship between the transmission rate and the performances, where the transmission rate $r$ is defined as the ratio between the number of channel uses $n$ and the dimension of each image data $k$.
For CIFAR 10 dataset, the transmission rate is computed as
\begin{equation}
    r = \frac{n}{32\times32 \times3}.
\end{equation}
We range the rate $r$ from $\frac{1}{48}$ to $\frac{1}{3}$, and fix channel SNR at 0 dB.
As a remark, the fixed 128 channel use in the previous section is equivalent to a transmission rate of $\frac{1}{24}$.

In Fig.~\ref{rate} we plot the performance of classification and image recovery at different transmission rates with 16QAM.
First, it can be observed that the classification accuracy remains steady at around 87$\%$ across all transmission rates for all methods, while the PSNR can be increased by up to 6 dB when $r$ increases from $\frac{1}{48}$ to $\frac{1}{3}$.
This shows that the recovery of semantic information requires much less channel resources than the recovery of source data.
Second, for image recovery, the proposed method always maintains a performance advantage among all the digital benchmark methods and has a close performance with the analog method.
For example, when $r=\frac{1}{6}$, the JCM method outperforms the DeepJSCC-Q method by 0.4 dB, the NN method by 0.6 dB and outperforms the Uniform method by 1.7 dB.
At the same time, the analog method has an performance advantage of 0.7 dB over the proposed method.

Fig.~\ref{imageRecovery} illustrates visual examples of the image recovery for all methods. 
We set the SNR at 0 dB and use 16QAM modulation.
We can observe that compared to other digital methods, the recovered images of the proposed method are clearer with more discernible features and outlines, which shows its advantage in recovering the source data.

\subsection{Performances against Different Modulation Orders}

\begin{figure*}[t]
    \centering
        \begin{subfigure}[b]{0.32\textwidth}
          \centering
          \includegraphics[width=0.8\textwidth]{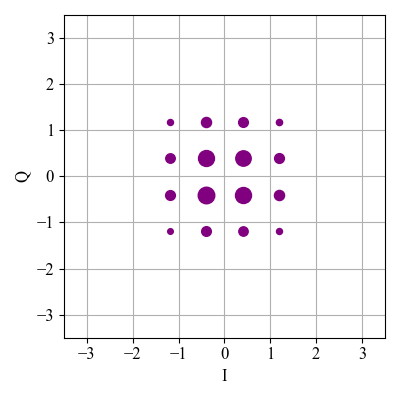}
          \caption{SNR = -12 dB.}
        \end{subfigure}
        \begin{subfigure}[b]{0.32\textwidth}
          \centering
          \includegraphics[width=0.8\textwidth]{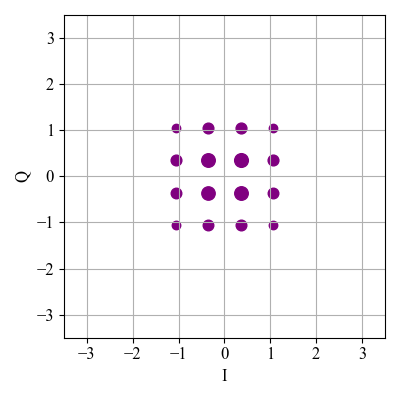}
          \caption{SNR = 0 dB.}
        \end{subfigure}
            \begin{subfigure}[b]{0.32\textwidth}
          \centering
          \includegraphics[width=0.8\textwidth]{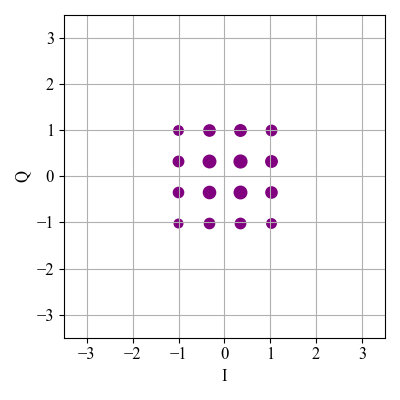}
          \caption{SNR = 18 dB.}
        \end{subfigure}
    \caption{The empirical distributions of the constellation points output by the JCM method with 16QAM at different channel SNRs. The sizes of the points are proportional to their probabilities of occurrence.}
    \label{distribution 16qam}
\end{figure*}

\begin{figure*}[t]
    \centering
        \begin{subfigure}[b]{0.32\textwidth}
          \centering
          \includegraphics[width=0.8\textwidth]{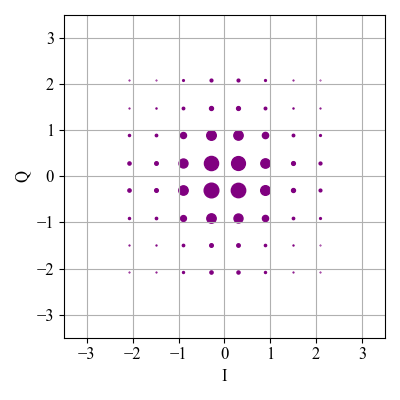}
          \caption{SNR = -12 dB.}
        \end{subfigure}
        \begin{subfigure}[b]{0.32\textwidth}
          \centering
          \includegraphics[width=0.8\textwidth]{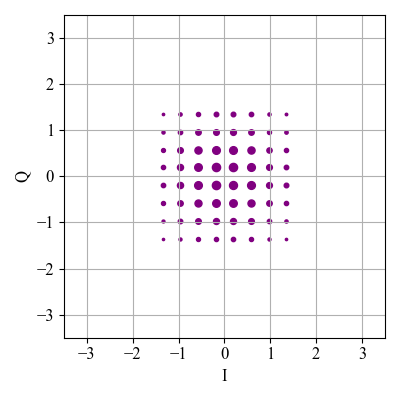}
          \caption{SNR = 0 dB.}
        \end{subfigure}
            \begin{subfigure}[b]{0.32\textwidth}
          \centering
          \includegraphics[width=0.8\textwidth]{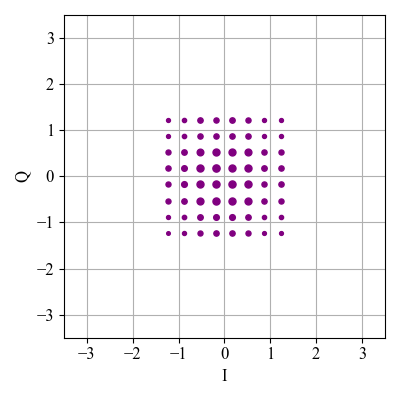}
          \caption{SNR = 18 dB.}
        \end{subfigure}
    \caption{The empirical distributions of the constellation points output by the JCM method with 64QAM at different channel SNRs. The sizes of the points are proportional to their probabilities of occurrence.}
    \label{distribution 64qam}
\end{figure*}

In previous subsections, we find that the performance of the proposed method improves with increasing modulation order. Therefore, in this subsection, we study more closely the effects of the modulation order on the proposed method.

In Fig.~\ref{64qam}, we plot the accuracy and PSNR performance of the proposed method versus the channel SNR in three rectangular M-QAM modulation: 4QAM, 16QAM and 64QAM. 
The transmission rate is set at $\frac{1}{24}$, which is equivalent to 128 channel uses.
First, in low SNR region, the three modulations 4QAM, 16QAM and 64QAM perform almost the same. 
In particular, when SNR $\le$ -12 dB, their curves of the classification accuracy and PSNR coincides.
As the channel SNR increases, higher order modulation begins to show advantages over lower order modulation.
At SNR $=$ 6 dB, with the same accuracy, 64QAM outperforms 16QAM by 0.3 dB in image recovery, and outperforms 4QAM by 1.8 dB.
These gaps respectively widens to 1.4 dB and 3.5 dB at SNR $=$ 18 dB.
Additionally, the performance gap between the analog method and the JCM method reduces with increasing modulation order.
For example, at SNR $=$ 18 dB, with the same accuracy, the gap in PSNR performance between 4QAM and the analog method is 4.8 dB.
This gap decreases to 2.8 dB and 1.4 dB respectively when the modulation changes to 16QAM and 64QAM.
As such, the JCM method provides a convenient way to conduct digital modulation in a semantic communication system, namely, by using the highest modulation order that the transmission device could support for the best performance it can offer.

To provide more insights into the JCM framework, in Fig.~\ref{distribution 16qam}, we plot the empirical distributions of the constellation points output by the JCM method with 16QAM at SNR $=$ -12 dB, 0 dB and 18 dB with the average power constrain $P=1$.
We can observe that the probability distribution has the appearance of a two-dimensional Gaussian distribution at low SNR and gradually changes towards a uniform distribution as the channel SNRs increase.
The same can be said of the empirical distribution of the 64QAM constellation points plotted in Fig.~\ref{distribution 64qam}.
This phenomenon coincides with the traditional work on probabilistic shaping in coded modulation\cite{ITonPS, 2019PS}, as well as the more recent work on the optimization of probabilistic shaping for physical layer using NNs\cite{NNforPS, NNforPS2}.
Probabilistic shaping approaches the Shannon limit by trading off bit rate with average power, resulting in larger Euclidean distance among constellation points at the same power level\cite{ITonPS, 2019PS}.
It is a well-known fact that the optimal shaping distribution for AWGN channels comes from the Maxwell-Boltzmann family\cite{ITonPS}.
Notably, JCM can approximately achieve this probabilistic shaping without the need for explicit instructions to the NNs regarding the probability distribution of the output signal, which underscores its impressive ability to match with the channel conditions.
\begin{figure*}[t]
    \centering
        \begin{subfigure}[b]{0.48\textwidth}
          \centering
          \includegraphics[width=1\textwidth]{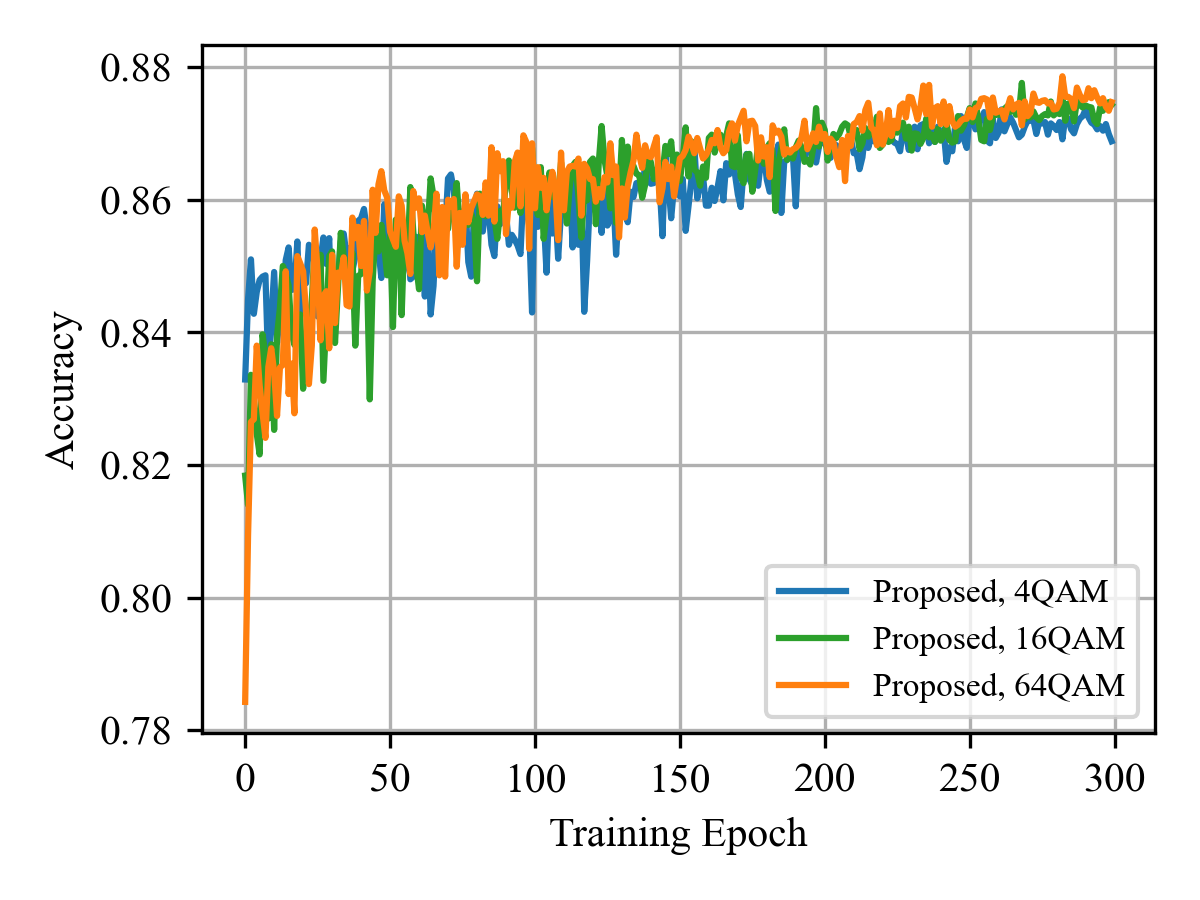}
          \caption{Convergence rate of the classification accuracy.}
          \label{acc rate}
        \end{subfigure}
        \begin{subfigure}[b]{0.48\textwidth}
          \centering
          \includegraphics[width=1\textwidth]{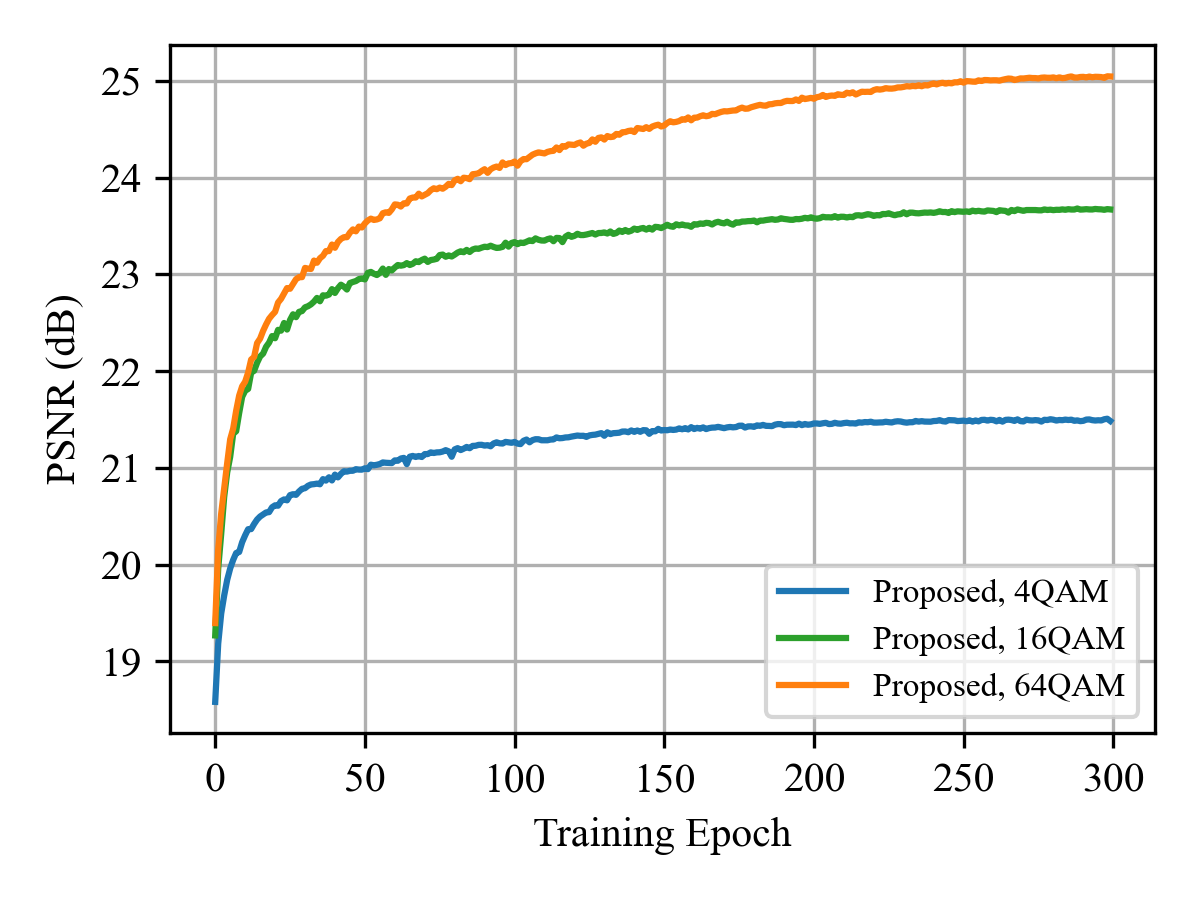}
          \caption{Convergence rate of PSNR.}
          \label{psnr rate}
        \end{subfigure}
    \caption{Convergence rates of the three modulation schemes. Channel use is set at 128. SNR is set at 18 dB.}
    \label{training rate}
\end{figure*}

To further find out whether higher order modulation will slow down the training, we plot in Fig.~\ref{training rate} the convergence rates when training the JCM framework with 4QAM, 16QAM and 64QAM, and also list in Table \ref{paras} their specific numbers of floating point operations (FLOPs), the number of parameters as well as the per-epoch training time.
It can be observed that the convergence rate of 64QAM is almost the same as that of lower order modulation, though bigger NNs with more parameters are needed.
These results validate the training efficiency and scalability of the JCM framework.

\begin{table}[t]
  \begin{center}
    \caption{The number of FLOPs, the number of parameters and training time of the JCM method for different modulation methods.}
    \begin{tabular}{|c|c|c|c|}
    \hline
    Modulation & 4QAM & 16QAM & 64QAM\\
    \hline
      FLOPs (G) & 0.643 & 0.651 & 0.668 \\
      \hline
      Params (M) &  24.6 & 33.0 & 49.8\\
      \hline
      Training Time for 1 Epoch (s) & 83 & 96 & 110 \\
      \hline
    \end{tabular}
    \label{paras}
  \end{center}
\end{table}

\section{Conclusion}
In this paper, we propose a joint coding-modulation framework for digital semantic communications based on the VAE architecture. In the JCM framework, the coding and modulation are jointly designed and learned, and thus are able to match with different channel conditions. 
Following this design, we successfully solve the non-differentiability problem of the digital modulation process using probabilistic models under the VAE architecture. 
Furthermore, aiming at the JCM framework, we employ variational inference and derive a loss function that can optimize the performance of both source reconstruction and semantic task execution.

Extensive experimental results validate the performance advantage of the JCM framework. 
It is found that the JCM framework outperforms separate design of semantic coding and modulation under various channel conditions, transmission rates, and modulation orders.
Furthermore, its performance gap to analog semantic communication reduces as the modulation order increases while enjoying the hardware implementation convenience.
Additionally, the JCM framework in higher order modulation achieves a better performance than in lower order modulation. Therefore, the JCM framework provides a convenient way of conducting digital modulation, namely by using the highest modulation order that the transmission device can support.

As a final note, this work has investigated how to perform digital modulation using existing modulation methods, such as BPSK, 4QAM and 16QAM, as well as their probabilistic shaping. 
Looking ahead, our future efforts will focus on the design and optimization of constellation geometry shaping to enhance the system's performance.

\begin{appendices}

\section{Proof of Theorem \ref{th1}}
To prove Theorem \ref{th1}, we first consider the decoding block, then merge the encoder-modulator into MI-OBJ. Specifically, for fixed transition probability, we want to find the best decoders in the sense that the NN parameterized probabilities $p_{de,s}(\mathbf{ s}|\mathbf{\hat z}, \phi)$ and $p_{de,o}(\mathbf{ x}|\mathbf{\hat z}, \psi)$ are closest to the true posterior distributions. Therefore, we use variational inference to derive a lower bound of MI-OBJ so that in the process of maximizing this lower bound, the NN parameterized probabilities can approach the true posterior distributions.
This lower bound is presented in Lemma \ref{lemmaVILB}.


\begin{lemma} A simpler lower bound of MI-OBJ is given by \eqref{lemma}
\begin{align}
I_{\theta}(\mathbf{S};\mathbf{\hat Z}) &+ \lambda\cdot I_{\theta}(\mathbf{\mathbf{X};\hat Z}) \ge
\mathbb{E}_{p(\mathbf{\hat z})}\lbrace \mathbb{E}_{p(\mathbf{s}|\mathbf{\hat z})} \log p_{de,s}(\mathbf{ s}|\mathbf{\hat z}, \phi)  
\nonumber\\&\ \ \ \ \ \ \ +\lambda\cdot \mathbb{E}_{p(\mathbf{x}|\mathbf{\hat z})} \log p_{de,o}(\mathbf{ x}|\mathbf{\hat z}, \psi)\rbrace + K
\label{lemma}
\end{align}
where $K=H(\mathbf{S})+\lambda\cdot H(\mathbf{X})$. 
\label{lemmaVILB}
\end{lemma}

\begin{proof}
Considering the term $I_{\theta}(\mathbf{S};\mathbf{\hat Z})$, we first expand it then add the NN parameterized approximate probability $p_{de,s}(\mathbf{ s}|\mathbf{\hat z}, \phi)$ into the equation. We split the term and get a KL divergence term between the true posterior distribution $p(\mathbf{s}|\mathbf{\hat{z}})$ and the approximate posterior distribution $p_{de,s}(\mathbf{ s}|\mathbf{\hat z}, \phi)$. Since the KL divergence is always non-negative, we can obtain a lower bound of $I_{\theta}(\mathbf{S};\mathbf{\hat Z})$. Specifically, we have
\begin{align}
    I_{\theta}(\mathbf{S};\mathbf{\hat Z})
    &=-H(\mathbf{S}|\mathbf{\hat Z}) + H(\mathbf{S}) \nonumber\\
    &=\mathbb{E}_{p(\mathbf{s},\mathbf{\hat z})} \log p(\mathbf{s}|\mathbf{\hat z}) + H(\mathbf{S}) \nonumber\\
    &=\mathbb{E}_{p(\mathbf{s},\mathbf{\hat z})} \log \left \lbrack p(\mathbf{s}|\mathbf{\hat z}) \frac{p_{de,s}(\mathbf{ s}|\mathbf{\hat z}, \phi)}{p_{de,s}(\mathbf{ s}|\mathbf{\hat z}, \phi)}\right \rbrack + H(\mathbf{S}) \nonumber\\
    &\overset{(a)}{=}\mathbb{E}_{p(\mathbf{\hat z})}\mathrm{KL}\left \lbrack p(\mathbf{s}|\mathbf{\hat z})||p_{de,s}(\mathbf{ s}|\mathbf{\hat z}, \phi)\right \rbrack \nonumber \\
    &\ \ \ \ +\mathbb{E}_{p(\mathbf{s},\mathbf{\hat z})} \log p_{de,s}(\mathbf{ s}|\mathbf{\hat z}, \phi) + H(\mathbf{S}) \nonumber \\
    &\overset{(b)}{\ge} \mathbb{E}_{p(\mathbf{s},\mathbf{\hat z})} \log p_{de,s}(\mathbf{ s}|\mathbf{\hat z}, \phi) + H(\mathbf{S})\nonumber \\
    &=\mathbb{E}_{p(\mathbf{\hat z})}\mathbb{E}_{p(\mathbf{s}|\mathbf{\hat z})} \log p_{de,s}(\mathbf{ s}|\mathbf{\hat z}, \phi) + H(\mathbf{S}),
    \label{sz}
\end{align}
where $(a)$ follows the definition of the KL divergence, and $p_{de,s}(\mathbf{ s}|\mathbf{\hat z}, \phi)$ is expanded as a variational approximation to the true posterior distribution $p(\mathbf{s}|\mathbf{\hat z})$; $(b)$ follows the non-negativity of the KL divergence.

Similarly, for $I_{\theta}(\mathbf{X};\mathbf{\hat Z})$ we have
\begin{align}
       I_{\theta}(\mathbf{X};\mathbf{\hat Z})
    &=-H(\mathbf{X}|\mathbf{\hat Z}) + H(\mathbf{X}) \nonumber\\
    &=\mathbb{E}_{p(\mathbf{x},\mathbf{\hat z})} \log p(\mathbf{x}|\mathbf{\hat z}) + H(\mathbf{X}) \nonumber\\
    &=\mathbb{E}_{p(\mathbf{x},\mathbf{\hat z})} \log \left \lbrack p(\mathbf{x}|\mathbf{\hat z}) \frac{p_{de,o}(\mathbf{ x}|\mathbf{\hat z}, \psi)}{p_{de,o}(\mathbf{ x}|\mathbf{\hat z}, \psi)}\right \rbrack + H(\mathbf{X}) \nonumber\\
    &=\mathbb{E}_{p(\mathbf{\hat z})}\mathrm{KL}\left \lbrack p(\mathbf{x}|\mathbf{\hat z})||p_{de,o}(\mathbf{ x}|\mathbf{\hat z}, \psi)\right \rbrack \nonumber \\
    &\ \ \ \ +\mathbb{E}_{p(\mathbf{x},\mathbf{\hat z})} \log p_{de,o}(\mathbf{ x}|\mathbf{\hat z}, \psi) + H(\mathbf{X}) \nonumber \\
    &\ge \mathbb{E}_{p(\mathbf{x},\mathbf{\hat z})} \log p_{de,o}(\mathbf{ x}|\mathbf{\hat z}, \psi) + H(\mathbf{X}) \nonumber\\
    &=\mathbb{E}_{p(\mathbf{\hat z})}\mathbb{E}_{p(\mathbf{x}|\mathbf{\hat z})} \log p_{de,o}(\mathbf{ x}|\mathbf{\hat z}, \psi) + H(\mathbf{X}).
     \label{xz}
\end{align}


Putting \eqref{sz} and \eqref{xz} together finishes the proof of Lemma \ref{lemmaVILB}.
\end{proof}

Then, We further prove Theorem 1 by expanding the probability of the received sequence $p(\mathbf{\hat z})$. By the Markov Chain $\mathbf{S}\to\mathbf{X}\to\mathbf{Z}\to\mathbf{\hat{Z}}\to(\mathbf{\hat S},\mathbf{\hat{X}})$, the probability of $\mathbf{\hat z}$ is determined by the NN parameterized transition probability of the encoder-modulator
\begin{align}
    p(\mathbf{\hat z})=&
    \int _{\mathbf{s},\mathbf{x},\mathbf{z}} p(\mathbf{s},\mathbf{x},\mathbf{z}, \mathbf{\hat z}) d\mathbf{s}d\mathbf{x}d\mathbf{z} \nonumber \\
    =&\int _{\mathbf{s},\mathbf{x},\mathbf{z}} p(\mathbf{s},\mathbf{x})p_{en}(\mathbf{z}|\mathbf{x}, \theta)p(\mathbf{\hat z|\mathbf{z}})d\mathbf{s}d\mathbf{x}d\mathbf{z}. \label{pz}
\end{align}
Then, we introduce \eqref{pz} into Lemma \ref{lemmaVILB}, getting
\begin{align}
&\ \ \ \ I_{\theta}(\mathbf{S};\mathbf{\hat Z}) + \lambda\cdot I_{\theta}(\mathbf{\mathbf{X};\hat Z})\nonumber \\ &\ge
\mathbb{E}_{p(\mathbf{\hat z})}\lbrace \mathbb{E}_{p(\mathbf{s}|\mathbf{\hat z})} \log p_{de,s}(\mathbf{ s}|\mathbf{\hat z}, \phi) \nonumber\\ 
&\ \ \ \ +\lambda\cdot \mathbb{E}_{p(\mathbf{x}|\mathbf{\hat z})} \log p_{de,o}(\mathbf{ x}|\mathbf{\hat z}, \psi)\rbrace + K \nonumber\\
&=
\mathbb{E}_{p(\mathbf{\hat z}|\mathbf{z})}\mathbb{E}_{p_{en}(\mathbf{z}|\mathbf{x}, \theta)}\mathbb{E}_{p(\mathbf{s},\mathbf{x})} \lbrace \mathbb{E}_{p(\mathbf{s}|\mathbf{\hat z})} \log p_{de,s}(\mathbf{s}|\mathbf{\hat z}, \phi)  \nonumber \\
&\ \ \ \ +\lambda\cdot \mathbb{E}_{p(\mathbf{x}|\mathbf{\hat z})} \log p_{de,o}(\mathbf{ x}|\mathbf{\hat z}, \psi)\rbrace +K,
\end{align}
which finishes the proof of Theorem \ref{th1}. Thereby, we get the variational inference lower bound with encoder parameters $\theta$ and decoder parameters $\phi$ and $\psi$.

\section{Proof of Proposition \ref{bpskCorollary}}

\begin{figure*}[t]
\begin{align}
\setcounter{equation}{28}
    I_{\theta}(\mathbf{X};\mathbf{\hat Z})&=-H(\mathbf{X}|\mathbf{\hat Z})+H(\mathbf{X}) \nonumber\\ 
    &\overset{(a)}{\ge} \mathbb{E}_{p(\mathbf{\hat z})}\mathbb{E}_{p(\mathbf{x}|\mathbf{\hat z})} \log p_{de,o}(\mathbf{ x}|\mathbf{\hat z}, \psi) + H(\mathbf{X}) \nonumber\\
    &\overset{(b)}{=} \int_{\mathbf{\hat{z}}}p(\mathbf{\hat z}) d \mathbf{\hat z}\int_{\mathbf{x}}\frac{1}{(2\pi )^{\frac{k}{2}} \sigma_{1}}e^{-\frac{1}{2}(\mathbf{x}-\mathbf{\boldsymbol{\mu}})^{T}\frac{1}{\sigma_1^{2}}\mathbf{I}(\mathbf{x}-\mathbf{\boldsymbol{\mu}})}\log \frac{1}{(2\pi )^{\frac{k}{2}} \sigma_{2}}e^{-\frac{1}{2}(\mathbf{ x}-f_\psi(\mathbf{\hat{z}}))^{T}\frac{1}{\sigma_2^{2}}\mathbf{I}(\mathbf{ x}-f_\psi(\mathbf{\hat{z}}))}d \mathbf{x}\nonumber \\&\ \ \ + H(\mathbf{X})\nonumber \\
    &\overset{(c)}{=} \int_{\mathbf{\hat{z}}}p(\mathbf{\hat z})  (b-d\left\lbrack(\boldsymbol{\mu}-f_\psi(\mathbf{\hat{z}}))^T(\boldsymbol{\mu}-f_\psi(\mathbf{\hat{z}}))\right\rbrack d \mathbf{\hat z} + H(\mathbf{X}) \nonumber\\
    &=-d \cdot \mathbb{E}_{p(\mathbf{\hat z})}
    ||\boldsymbol{\mu}-f_\psi(\mathbf{\hat{z}})||_2^2
    +b + H(\mathbf{X}),
    \label{mutualX}
\end{align}
\hrulefill
\end{figure*}

For BPSK modulation, the categorical distribution of $Z_i$ is degraded to a Bernoulli distribution. We denote the two constellation symbols as $c_1=1$ and $c_2=-1$. Therefore, the transition probability of $\mathbf{Z}$ can be expressed as 
\begin{align}
\setcounter{equation}{25}
    p_{en}(\mathbf{z}|\mathbf{x}, \theta)&=\textstyle\prod_{i=1}^{n}p(z_i|\mathbf{x}, \theta)\nonumber\\
    &=\textstyle\prod_{i=1}^{n} \textstyle\prod_{m=1}^{2}( q_{im|\mathbf{x},\theta})^{\mathbb{I}\left\lbrace z_i=c_m \right\rbrace}\nonumber\\
    &=\textstyle\prod_{i=1}^{n} ( q_{i1|\mathbf{x},\theta})^{\mathbb{I}\left\lbrace z_i=1 \right\rbrace}( q_{i2|\mathbf{x},\theta})^{\mathbb{I}\left\lbrace z_i=-1 \right\rbrace}\nonumber\\
    &=\textstyle\prod_{i=1}^{n}( q_{i|\mathbf{x},\theta}) ^{\frac{1}{2}(1+z_i)}(1- q_{i|\mathbf{x},\theta}) ^{\frac{1}{2}(1-z_i)},\label{bpskOneSymbol}
\end{align}
where to simplify the notation, we denote the probability $q_{i1|\mathbf{x},\theta}$ as $q_{i|\mathbf{x},\theta}$ and correspondingly $q_{i2|\mathbf{x},\theta}$ as $( 1-q_{i|\mathbf{x},\theta})$. 

\section{Proof of Proposition \ref{qamCorollary}}

For rectangular M-QAM where $M=2^{2a}, a=1, 2, ...$, we also consider the I channel and Q channel to be conditionally independent. 
Each constellation symbol $c_m\in \cal C$, $m=1,2,...,M$, can be expressed as $c_{Ir}+j\cdot c_{Qs}$, where we use $j$ to denote the imaginary unit, $c_{Ir}$ to denote the amplitude of the I channel with $c_{Ir}= \frac{2r+1}{\sqrt{M}-1}, r=\frac{-\sqrt{M}}{2},-\frac{\sqrt{M}}{2}+1,...,\frac{\sqrt{M}}{2}-1$, and $c_{Qs}$ to denote the amplitude of the Q channel with $c_{Qs} = \frac{2s+1}{\sqrt{M}-1}, s=\frac{-\sqrt{M}}{2},-\frac{\sqrt{M}}{2}+1,...,\frac{\sqrt{M}}{2}-1$.
    
Therefore, the probability $P(Z_i=c_m|\mathbf{x}, \theta)$ can be written as
\begin{align}
        P(Z_i=c_m|\mathbf{x}, \theta)&=P(Z_{Ii} + j\cdot Z_{Qi} = c_{Ir}+j\cdot c_{Qs}|\mathbf{x}, \theta) \nonumber\\
        &=P(Z_{Ii}=c_{Ir}|\mathbf{x}, \theta)\cdot p(Z_{Qi}=c_{Qs}|\mathbf{x}, \theta)\nonumber\\
        &=q_{ir|\mathbf{x},\theta}^{I}\cdot q_{is|\mathbf{x},\theta}^{Q},
        \label{mqam1}
\end{align}
where to simplify the notation, we use $q_{ir|\mathbf{x},\theta}^{I}$ to denote $P(Z_{Ii}=c_{Ir}|\mathbf{x}, \theta)$ and $q_{is|\mathbf{x},\theta}^{Q}$ to denote $P(Z_{Qi}=c_{Qs}|\mathbf{x}, \theta)$.
Correspondingly, $\mathbb{I}\left\lbrace z_i=c_m \right\rbrace$ can be written as
\begin{align}
        \mathbb{I}\left\lbrace z_i=c_m \right\rbrace
        &=\mathbb{I}\left\lbrace z_{Ii} + j\cdot z_{Qi} = c_{Ir}+j\cdot c_{Qs} \right\rbrace\nonumber\\
        &=\mathbb{I}\left\lbrace z_{Ii}=c_{Ir} \right\rbrace \cdot \mathbb{I}\left\lbrace z_{Qi}=c_{Qs}\right\rbrace\nonumber\\
        &=\mathbb{I}\left\lbrace z_{Ii}=\frac{2r+1}{\sqrt{M}-1} \right\rbrace \cdot \mathbb{I}\left\lbrace z_{Qi}=\frac{2s+1}{\sqrt{M}-1}\right\rbrace.
        \label{mqam2}
\end{align}
Substituting \eqref{mqam1} and \eqref{mqam2} into \eqref{generalpmf} obtains \eqref{mqam}.

\section{Proof of Corollary \ref{corollaryImage}}

In image semantic communications, for the source data, we take the Gaussian assumptions \eqref{true posterior} and \eqref{variational} into consideration, and further obtain \eqref{mutualX} at the top of this page,
where in $(a)$ $p_{de,o}(\mathbf{ x}|\mathbf{\hat z}, \psi)$ is the NN parameterized variational approximation to the true posterior distribution $p(\mathbf{x}|\mathbf{\hat z})$; $(b)$ follows the expansion of the expectation; $(c)$ follows the repeated use of the equation $\mathbf{x}^{\mathrm{T}}\mathrm{A}\mathbf{x}=\mathrm{tr}(\mathrm{A}\mathbf{x}\mathbf{x}^{\mathrm{T}})$ with $d$ and $b$ being two constants.

As for the semantic information, since it is set as the classification label of the images, it is a discrete variable with a categorical distribution. Therefore, we have
\begin{align}
\setcounter{equation}{29}
    I_{\theta}(\mathbf{S};\mathbf{\hat Z})&=-H(\mathbf{S}|\mathbf{\hat Z})+H(\mathbf{S}) \nonumber\\ 
    &\ge \mathbb{E}_{p(\mathbf{\hat z})}\mathbb{E}_{p(\mathbf{s}|\mathbf{\hat z})} \log p_{de,s}(\mathbf{ s}|\mathbf{\hat z}, \psi) + H(\mathbf{S}) \nonumber\\
    &=\mathbb{E}_{p(\mathbf{\hat z})}\Biggl(\sum^{L}_{l=1}p(\mathbf{s}=l|\mathbf{\hat{z}})\log p_{de,s}(\mathbf{ s}=l|\mathbf{\hat z}, \phi)\Biggr)\nonumber\\&\ \ \ \ \ +H(\mathbf{S}).
\end{align}

Then, by replacing probability distributions with empirical distributions, we have
\begin{align}
&\lim_{N \to \infty} \frac{1}{N}\sum_{n=1}^{N} \Biggl(  \sum_{l=1}^{L} p(\mathbf{s}_n=l|\mathbf{\hat{z}}_\theta^{n})\log p_{de,s}(\mathbf{ s}_n=l|\mathbf{\hat z}_\theta^{n}, \phi) \Biggr)\nonumber\\
  &\ \ =  \mathbb{E}_{p(\mathbf{\hat z})}\Biggl(\sum^{L}_{l=1}p(\mathbf{s}=l|\mathbf{\hat{z}})\log p_{de,s}(\mathbf{ s}=l|\mathbf{\hat z}, \phi)\Biggr), \label{emp1} \\
&\lim_{N \to \infty} \frac{1}{N}\sum_{n=1}^{N}  ||\mathbf{x}_n- f_\psi(\mathbf{\hat z}_\theta^{n})||_2^2 = \mathbb{E}_{p(\mathbf{\hat z})}||\boldsymbol{\mu}-f_\psi(\mathbf{\hat{z}})||_2^2.\label{emp2}
\end{align}
Taking \eqref{emp1} and \eqref{emp2} into VILB \eqref{elbo}, we get the loss function for image semantic communications and complete the proof of Corollary \ref{corollaryImage}.




\end{appendices}

\bibliographystyle{IEEEtran}
\bibliography{jcm_bib}{}

\end{document}